\documentclass[11pt,a4paper]{article}
\usepackage[T1]{fontenc}
\usepackage[utf8]{inputenc}
\usepackage{authblk}
\usepackage{fullpage}

\usepackage{version}
\usepackage{amssymb,amsmath,amscd,amsfonts,latexsym}
\usepackage{tikz}
\usetikzlibrary{arrows,automata,backgrounds,decorations}
\usetikzlibrary{shapes.multipart}
\usepgflibrary{decorations.pathreplacing} 
\usepackage{graphicx}
\usepackage{todonotes}
\usepackage{tabularx}
\usetikzlibrary{arrows,calc,topaths}
\tikzstyle{small}=[font=\footnotesize]
\tikzset{
    every picture/.style={>=stealth,auto,node distance=2cm},
}
\usepackage{centernot}

\newcommand{\Act}{\mathit{Act}}
\newcommand{\CA}{\mathit{CA}}
\newcommand{\CB}{\mathit{CB}}
\newcommand{\OCN}{OCN}
\newcommand{\OCA}{OCA}
\newcommand{\ignore}[1]{}
\newcommand{\WSIM}[2]{\ensuremath{\,\mwsim^{#1}_{#2}}\,}
\newcommand{\notWSIM}[2]{{}\not\mwsim^{#1}_{#2}{}}
\newcommand{\ptnote}[1]{{}\todo[backgroundcolor=blue!40,inline,size=\small]{#1}}

\newcommand{\N}{\mathbb{N}}

\newcommand{\Z}{\mathbb{Z}}

\newcommand{\x}{\times}

\newcommand{\R}{Spoiler}
\newcommand{\V}{Duplicator}

\newenvironment{proof}[0]{\emph{Proof}} {\hfill $\Box$ \smallskip }

\newcommand{\step}[1]{\Step{#1}{}{}}

\newcommand{\wstep}[1]{\Wstep{#1}{}{}}
\newcommand{\notwstep}[1]{\notWstep{#1}{}{}}
\newcommand{\Wstep}[3]{\ensuremath{\,{\stackrel{#1}{\Longrightarrow}}\!{}^{\scriptstyle{#2}}_{\scriptstyle{#3}}}\,}
\newcommand{\notWstep}[3]{\ensuremath{\,{\centernot{\stackrel{#1}{\Longrightarrow}}\!{}}^{\scriptstyle{#2}}_{\scriptstyle{#3}}}\,}
\newcommand{\Step}[3]{\ensuremath{\,{\stackrel{#1}{\longrightarrow}}\!{}^{\scriptstyle{#2}}_{\scriptstyle{#3}}}\,}

\newcommand{\msim}{\SIM{}{}}
\newcommand{\SIM}[2]{\ensuremath{\,\preceq^{#1}_{#2}}\,}
\newcommand{\notSIM}[2]{{}\not\preceq^{#1}_{#2}{}}
\newcommand{\mwsim}{\leqq}

\newcommand{\Ord}{\ensuremath{Ord}}

\newcommand{\nic}[1]{}

\newcommand{\eps}{\varepsilon}

\newcommand{\Alf}{A}

\newtheorem{theorem}{{\bf Theorem}}

\newtheorem{definition}[theorem]{{\bf Definition}}
\newtheorem{example}[theorem]{{\bf Example}}
\newtheorem{lemma}[theorem]{{\bf Lemma}}
\newtheorem{proposition}[theorem]{{\bf Proposition}}
\newtheorem{remark}[theorem]{{\bf Remark}}

\newcommand{\qed}{}

\begin{document}
\title{\bf Decidability of Weak Simulation on One-counter Nets\thanks{
Technical Report EDI-INF-RR-1415 of the School of Informatics at the University
of Edinburgh, UK. (http://www.inf.ed.ac.uk/publications/report/).
Extended version of material presented at LICS 2013.
Made available at arXiv.org - Non-exclusive license to distribute
(http://arxiv.org/licenses/nonexclusive-distrib/1.0/).
}}

\date{}

\author[1]{Piotr Hofman}
\author[2]{Richard Mayr}
\author[2]{Patrick Totzke}
\affil[1]{University of Warsaw, Poland}
\affil[2]{University of Edinburgh, UK}

\maketitle

\begin{abstract}
\noindent
One-counter nets (\OCN) are Petri nets with exactly one unbounded place.
They are equivalent to a subclass of one-counter automata with
only a weak test for zero. 

We show that weak simulation preorder is decidable for \OCN\
and that weak simulation approximants do not converge
at level $\omega$, but only at $\omega^2$.
In contrast, other semantic relations like weak bisimulation 
are undecidable for \OCN\ \cite{May2003}, and so are weak (and strong) trace inclusion
(Sec.~\ref{sec:traces}).

\end{abstract}


\section{Introduction}

{\bf\noindent The model.}
One-counter automata (\OCA) are Minsky counter automata with only one counter,
and they can also be seen as a subclass of pushdown automata with just one
stack symbol (plus a bottom symbol).
One-counter nets (\OCN) are Petri nets with exactly one unbounded place,
and they correspond to a subclass of OCA where the counter
cannot be fully tested for zero (i.e., transitions enabled at counter value
zero are also enabled at nonzero counter values).
\OCN\ are arguably the simplest model of discrete infinite-state systems,
except for those that do not have 
a global finite control, e.g., (commutative) context-free grammars.

{\bf\noindent The state of the art.}
Verification problems for \OCA\ and \OCN\ have been extensively studied,
particularly model checking problems with various temporal logics and
semantic preorder/equivalence checking w.r.t. given notions of behavior. 
$\mu$-calculus model checking \cite{Ser2006,GMT2009} and CTL model checking
\cite{Goller:STACS2010} are PSPACE-complete for \OCA/\OCN,
while EF model checking is $P^{\it NP}$-complete \cite{GMT2009}.
There are many notions of semantic equivalences \cite{Gla2001},
but the most common ones are the following (ordered from finer to coarser): 
bisimulation, simulation and trace equivalence.
Each of these have their standard (called strong) variant, and a weak variant
that abstracts from arbitrarily long sequences of internal actions.

Strong bisimulation for 
\OCA/\OCN\ is PSPACE-complete \cite{BGJ2010}. However, \OCA\ and \OCN\ differ
w.r.t. strong simulation. While strong simulation is decidable for \OCN\
\cite{AC1998,JM1999,JKM2000}, strong simulation and trace inclusion are 
undecidable for \OCA\ \cite{JMS1999}.
((Bi)simulation checking is also undecidable for models with more than one
unbounded counter/place \cite{Jancar:PN-bisimilarity-TCS}.)

The situation changes when one considers weak semantic equivalences/preorders
that abstract from internal actions.
One can define upper-approximations of
(bi)simulation up-to $k$ by considering only $k$ rounds in the
(bi)simulation game. For strong (bi)simulation, these $k$-approximants
converge to (bi)simulation at level $k=\omega$, provided that the systems are finitely branching.
This is not the case for weak (bi)simulation.
Here the approximants are guaranteed to converge only at higher ordinals,
due to the implicit infinite branching capability introduced by
the abstraction.
This is why it is so hard to prove
semi-decidability of weak non-(bi)simulation for many classes of
infinite-state transition systems.

For \OCA/\OCN\ it was shown that weak bisimulation is undecidable
\cite{May2003}. Moreover, weak (and strong) simulation and trace inclusion are undecidable
for \OCA\ \cite{JMS1999,Valiant1973}.
However, it was an open question whether weak simulation is decidable for
\OCN.
Moreover, the decidability of strong and weak trace inclusion was open for
\OCN\,\cite{EN94}.

{\bf\noindent Our contribution.}
We show that weak simulation preorder is decidable for \OCN.
In fact, the weak simulation relation on \OCN\ is effectively semilinear.
Moreover, we show that weak simulation approximants only converge at level $\omega^2$
on \OCN.
The decidability of weak simulation is in contrast to the undecidability of
weak bisimulation for \OCN\ \cite{May2003}. This is surprising, because it
goes against a common trend. On almost all other classes of systems,
bisimulation problems are computationally easier than the corresponding
simulation problems \cite{KM:Why-bisim-sim}.

On the other hand, we show that strong and weak trace inclusion are
undecidable even for \OCN.

Finally, we study checking strong and weak (bi)simulation
and trace inclusion between \OCA/\OCN\ and finite systems, and close some
remaining complexity gaps in this area.

\section{Preliminaries}\label{sec:prelim}

    \paragraph{One-counter systems.}
    We consider infinite-state labeled transition systems induced by \OCA\ and /\OCN, respectively.
    A \emph{labeled transition system} is described by a triple
    $T=(V,\Act,\step{})$ where $V$ is a (possibly infinite) set of
    states, $\Act$ is a finite set of action labels and $\step{}\subseteq V\x
    \Act\x V$ is the labeled transition relation. We write $\step{}^*$ for the transitive and reflexive closure of
    $\step{}$ and use the infix notation $s\step{a}s'$ for a
    transition $(s,a,s')\in\step{}$, in which case we say $T$ makes an \emph{$a$-step} from $s$ to
    $s'$. 
    Given a finite or infinite 
    sequence of symbols $w \in \Act^*$ or $w \in \Act^\omega$ resp.,
    we write $|w|\in\N\cup\{\omega\}$ for the length
    of $w$. If $w$ is finite, we denote its $i$-fold concatenation by $w^i$.

    \begin{definition}[One-Counter Automata / Nets]
        A \emph{one-counter automaton} $A=(Q,\Act,\delta,\delta_0)$
        is given by a finite set
        of control-states $Q$, a finite set of action labels $\Act$
        and transition relations $\delta\subseteq Q\x \Act\x\{-1,0,1\}\x Q$
        and $\delta_0\subseteq Q\x \Act\x\{0,1\}\x Q$.
        It induces an infinite-state labeled transition system over the
        stateset $Q\x\N$, whose elements will be written as $pm$, as follows.
        $pm\step{a}p'm'$ iff
        \begin{enumerate}
          \item $(p,a,d,p')\in\delta$ and $m'=m+d\ge0$ or
          \item $(p,a,d,p')\in\delta_0$, $m=0$ and $m'=d$.
        \end{enumerate}
        Such an automaton is called a \emph{one-counter net} if
        $\delta_0=\emptyset$, i.e., if the automaton 
        cannot test if the counter is equal to $0$.
    \end{definition}

    \paragraph{Weak Simulation.}
    In a weak semantics, one needs to abstract from internal actions.
    Thus one assumes a dedicated action $\tau\in \Act$ that
    is used to model internal non-observable steps and defines the \emph{weak step} relation
    $\wstep{}$ by 
    \begin{align*}
        \wstep{\tau} &=\Step{\tau}{*}{} \mbox{, and for } a\neq\tau,
        &\wstep{a} =\Step{\tau}{*}{}\cdot\step{a}\cdot\Step{\tau}{*}{}
    \end{align*}
    
    Simulation and weak simulation are semantic preorders in van Glabbeeks
    linear time -- branching time spectrum \cite{Gla2001}, which are used to compare the behavior of processes.
    Their standard co-inductive definition is as follows.
    A binary relation $R\subseteq V^2$ on the states of a labeled transition
    system is a \emph{simulation} if $sRt$ implies that
    for all $s\step{a}s'$ there is a $t'$ such that $t\step{a}t'$ and $s'Rt'$. 
    Similarly, $R$ is a \emph{weak simulation} if $sRt$
    implies that for all $s\step{a}t$ there is a $t'$ such that $t\wstep{a}t'$ and
    $s'Rt'$. 
    (Weak) simulations {are closed under union, so there exists a
    unique maximal simulation $\msim$, resp. weak simulation $\mwsim$, which
    is a preorder on $V$.
    A (weak) bisimulation is a symmetric (weak) simulation.
    The maximal (weak) bisimulation is an equivalence.

    Simulation preorder can also be characterized in terms of ordinal \emph{approximant} relations
    $\SIM{}{\alpha}$, which are inductively defined as follows.
    $\SIM{}{0} = V^2$. For successors $\alpha+1$ let $s \SIM{}{\alpha+1} t$ iff for all $s\step{a}s'$
    there is a $t\step{a}t'$ such that $s'\SIM{}{\alpha}t'$. For limit ordinals $\lambda$
    define $\SIM{}{\lambda} = \bigcap_{\alpha<\lambda}\SIM{}{\alpha}$.

    This inductive notion of approximants can be interpreted as an interactive game between
    two players \R\ and \V, where the latter tries to stepwise match the moves of the
    former.  
    A \emph{play} is a finite or infinite sequence of pairs of transition system states.
    For a finite play $(E_0,F_0),(E_1,F_1),\dots,(E_i,F_i)$ the next pair $(E_{i+1},F_{i+1})$
    is determined by a round of choices: \R\ chooses a transition $E_i\step{a}E_{i+1}$,
    then \V\ responds by choosing an equally labeled transition $F_i\step{a}F_{i+1}$.
    A pair $(E,F)$ of states is \emph{directly winning} for \R\ if she can choose a transition
    $E\step{a}E'$ so that her opponent cannot respond, i.e. $\lnot\exists F'. F\step{a}F'$.
    A play is won by \R\ if a pair of states is reached that is directly winning for her, otherwise \V\ wins
    the play.
    A \emph{strategy} is a set of rules that tells the player which valid move to choose.
    A player plays according to a strategy if all her moves obey the rules of the strategy.
    A strategy is \emph{winning} from $(E,F)$ if every play that starts in $(E,F)$
    and which is played according to that strategy is winning.
    \begin{proposition}
        \label{prop:game-char}
        For any two states $(E,F)$ of a transition system $T$, \V\ has a winning
        strategy in $(E,F)$ in the simulation game
        iff $E\SIM{}{\alpha}F$ for all ordinals $\alpha$
        iff $E\msim F$.
    \end{proposition}
    Weak simulation approximants $\WSIM{}{\alpha}$ and games are defined analogously but allow \V\ to
    make weak steps and characterize $\mwsim$.

    \paragraph{The Problem.}
    We consider the problem of deciding weak simulation preorder on
    one-counter nets. An instance is given by a one-counter net $N=(Q,\Act,\delta)$
    and configurations $pm$ and $qn$, and the question is 
    whether $pm \mwsim qn$ holds. Generally, we want to compute a
    representation of the semilinear set $W(p,q)=\{(m,n)\,|\, pm \mwsim qn\}$.

%


\section{Reduction to Strong Simulation on $\omega$-Nets}
    
First we reduce the weak simulation problem on one-counter nets to
a strong simulation problem on a slightly generalized model that we call
$\omega$-nets.
In $\omega$-nets, there exist dedicated transitions with symbolic effect $\omega$, which allow to
arbitrarily increase the counter in a single step.
Checking weak simulation between two one-counter nets
can be reduced to strong simulation between a one-counter net and an $\omega$-net.
    
    \begin{definition}[$\omega$-Nets]
        An \emph{$\omega$-net} $N=(Q,Act,\delta)$ is given by a finite set of
        control-states $Q$, a finite set of actions $\Act$ and transitions
        $\delta\subseteq Q\x \Act\x\{-1,0,1,\omega\}\x Q$. It induces a transition system over the
        stateset $Q\x\N$
        that allows a step $pm\step{a}p'm'$ if either $(p,a,d,p')\in\delta$ and $m'=m+d\in\N$ or if
        $(p,a,\omega,p')\in\delta$ and $m' > m$.
    \end{definition}
    Every one-counter net is a $\omega$-net without $\omega$-transitions.
    Unlike one-counter nets, $\omega$-nets can yield infinitely branching
    transition systems, since each $\omega$-transition $(p,a,\omega,p')$ introduces steps 
    $pm\step{a}p'm'$ for any two naturals $m' > m$.

    It is easily verified that $\omega$-net (and hence also one-counter net) processes satisfy the following monotonicity property.

    \begin{proposition}[Monotonicity of $\msim$]\label{reduction:monotonicity}
        $pm\step{a}p'm'$ implies $p(m+d)\step{a}p'(m'+d)$ for all $d \in\N$.
        Moreover, $pm\msim qn$ implies $pm'\msim qn'$ for $m'\le m$, $n'\ge n$.
    \end{proposition}

    The following theorem justifies our focus on strong simulation games where \V\ plays
    on an $\omega$-net process.

    \begin{theorem}
        \label{thm:reduction}
        Checking weak simulation between two one-counter net processes can be reduced to
        checking strong simulation between a one-counter net process and an
        $\omega$-net process.
        Formally, for two one-counter nets $M$ and $N$ with states
        $Q_M$ and $Q_N$ resp., one
        can effectively construct a \OCN\ $M'$ with states $Q_{M'}\supseteq Q_M$
        and a $\omega$-net $N'$ with states $Q_{N'}\supseteq Q_N$ such that for
        each pair $p,q\in Q_M\x Q_N$ of original control states and
                any ordinal $\alpha$ the following hold.
        \begin{enumerate}
            \item $pm\mwsim qn$ w.r.t. $\!M,N$ iff $pm\msim\!qn$ w.r.t. $M',N'$.\label{thm:reduction:main}
            \item If $pm\WSIM{}{\alpha}\!qn$ w.r.t. $\!M,N$
                then $pm\SIM{}{\alpha}\!qn$ w.r.t. $M',N'$.\label{thm:reduction:ordinals}
        \end{enumerate}
    \end{theorem}
    \begin{proof}(Sketch.)
    The idea of the proof is to look for counter-increasing cyclic paths
    via $\tau$-labeled transitions in the control graph
    and to introduce $\omega$-transitions accordingly.
    For any path that reads a single visible action and visits a 'generator' state that is
    part of a silent cycle with positive effect, we add an $\omega$-transition.
    For all of the finitely many non-cyclic paths that read a single visible action
    we introduce direct transitions.
    A full proof is given in Appendix~\ref{app:reduction}.
    \end{proof}
\section{Approximants}\label{sec:approximants}
We generalize the notion of $\SIM{}{\alpha}$ simulation approximants in 
the case of simulation between one-counter and $\omega$-net processes.
This yields approximants that converge at a finite level for any pair of nets.

First we define approximants $\SIM{\beta}{\alpha}$ in two (ordinal) dimensions.
From a game-theoretic perspective the subscript $\alpha$ indicates the number of rounds \V\ can survive
and the superscript $\beta$ denotes the number of $\omega$-steps \R\ needs to allow.
E.g., $pm\SIM{2}{5}qn$ if \V\ can guarantee that no play of the
simulation game that contains $< 2$ $\omega$-steps
is losing for him in less than $6$ rounds.
If not stated otherwise we assume that $N=(Q,\Act, \delta)$ is a one-counter net
and $N'=(Q',\Act,\delta')$ is an $\omega$-net.

\begin{definition}\label{def:approximants}
    We define \emph{approximants} for ordinals $\alpha$ and $\beta$ as follows.
    Let $\SIM{0}{\alpha}=\SIM{\beta}{0} = Q\x\N\x Q'\x\N$, the full relation.
    For successor ordinals $\alpha+1,\beta+1$ let $pm \SIM{\beta+1}{\alpha+1} qn$ iff for all
    $pm\step{a}p'm'$ there is a step $qn\step{a}q'n'$ s.t. either
    \begin{enumerate}
        \item $(q,a,\omega, q')\in\delta'$ (the step is due to an $\omega$-transition) and
            $p'm'\SIM{\beta}{\alpha} q'n'$, or
        \item $(q,a,\omega, q')\notin\delta'$, but $(q,a,(n'-n), q')\in\delta'$
          (i.e., there is no $\omega$-transition and the step is due to a normal transition)
          and $p'm'\SIM{\beta+1}{\alpha} q'n'$.
    \end{enumerate}
    For limit ordinals $\lambda$ we define
    $\SIM{\lambda}{\alpha}= \bigcap_{\beta<\lambda}\SIM{\beta}{\alpha}$ and
    $\SIM{\beta}{\lambda} = \bigcap_{\alpha<\lambda}\SIM{\beta}{\alpha}$.
    Finally,
    \begin{align}\label{def:big_approximants}
        \SIM{\beta}{} = \bigcap_{\alpha\in\Ord} \SIM{\beta}{\alpha}
        &&\SIM{}{\alpha} = \bigcap_{\beta\in\Ord} \SIM{\beta}{\alpha}.
    \end{align}
\end{definition}

    $\SIM{}{\alpha}$ corresponds to the usual notion of simulation approximants
    and $\SIM{\beta}{}$ is a special notion derived from the syntactic peculiarity of
    $\omega$-transitions present in the game on one-counter vs.\ $\omega$-nets.

\begin{example}
    Consider a net that consists of a single $a$-labeled loop in state $X$ and the
    $\omega$-net with transitions $Y\step{a,\omega}Z\step{a,-1}Z$ only.
    We see that for any $m,n\in\N$,
    $Xm\SIM{}{n}Zn \,\not\succeq_{n+1}\,Xm$. Moreover,
    $Xm\SIM{}{\omega}Yn$ but $Xm\notSIM{}{\omega+1}Yn$ and
    $Xm\SIM{1}{}Yn$ but $Xm\notSIM{2}{\omega+1}Yn$ and thus $Xm\notSIM{2}{}Yn$.
\end{example}

\begin{definition}\label{def:approxgame}
    An \emph{approximant game} is played in rounds between \R\ and \V.
    Game positions are quadruples $(pm,qn,\alpha,\beta)$ where $pm, qn$ are configurations
    of $N$ and $N'$ respectively, and $\alpha,\beta$ are ordinals called step- and $\omega$-counter.
    In each round that starts in $(pm,qn,\alpha,\beta)$:
    \begin{itemize}
      \item \R\ chooses ordinals $\alpha'<\alpha$ and $\beta'<\beta$,
      \item \R\ makes a step $pm\step{a}p'm'$,
      \item \V\ responds by making a step $qn\step{a}q'n'$ using some transition $t$.
    \end{itemize}
    If $t$ was an $\omega$-transition the game continues from position $(p'm',q'n',\alpha',\beta')$,
    Otherwise the next round starts at $(p'm',q'n',\alpha',\beta)$
    (in this case \R's choice of $\beta'$ becomes irrelevant).
    If a player cannot move the other wins and if $\alpha$ or $\beta$ becomes $0$, \V\ wins.
\end{definition}
\begin{lemma}
    \label{approximants:game:monotonicity}
    If \V\ wins the approximation game from $(pm,qn,\alpha,\beta)$
    then he also wins the game from $(pm,qn,\alpha',\beta')$ for any $\alpha'\le \alpha$
    and $\beta'\le\beta$.
\end{lemma}
\begin{proof}
If \V\ has a winning strategy in the game from $(pm,qn,\alpha,\beta)$ then he can
use the same strategy in the game from
$(pm,qn,\alpha',\beta')$ and maintain the invariant that the pair of
ordinals in the game configuration is pointwise smaller than the pair in the
original game. Thus \V\ wins from $(pm,qn,\alpha',\beta')$. \qed
\end{proof}

\begin{lemma}
    \label{lem:game_interpretation}
    $pm\SIM{\beta}{\alpha} qn$ iff \V\ has a strategy to win the approximation game that starts in
    $(pm,qn,\alpha,\beta)$.
\end{lemma}
\begin{proof}
    We show both directions by well-founded induction on the pairs of ordinals
    $(\alpha,\beta)$.

    For the ``only if'' direction we assume $pm\SIM{\beta}{\alpha}qn$ and show that \V\ wins the game
    from $(pm,qn,\alpha,\beta)$.
    In the base case of $\alpha=0$ or $\beta=0$ \V\ directly wins by definition.
    By induction hypothesis we assume that the claim is true for all pairs
    pointwise smaller than $(\alpha,\beta)$.
    \R\ starts a round by picking ordinals $\alpha'<\alpha$ and $\beta'<\beta$
    and moves $pm\step{a}p'm'$.
    We distinguish two cases, depending on whether $\beta$ is a limit or
    successor ordinal.

    Case 1: $\beta$ is a successor ordinal.
    By Lemma~\ref{approximants:game:monotonicity} we can safely assume that
    $\beta'=\beta-1$.
    By our assumption $pm\SIM{\beta}{\alpha}qn$ and Def.~\ref{def:approximants},
    there must be a response $qn\step{a}q'n'$ that is either due to an $\omega$-transition
    and then $p'm'\SIM{\beta'}{\alpha'}q'n'$ or due to an ordinary transition, in which case we have
    $p'm'\SIM{\beta'+1}{\alpha'}q'n'$. In both cases, we know by the induction hypothesis
    that \V\ wins from this next position and thus also from the initial position.

    Case 2: $\beta$ is a limit ordinal.
    By $pm\SIM{\beta}{\alpha}qn$ and Def.~\ref{def:approximants},
    we obtain $pm\SIM{\gamma}{\alpha} qn\text{ for all }\gamma <\beta$.
    If $\alpha$ is a successor ordinal then, by Lemma~\ref{approximants:game:monotonicity},
    we can safely assume that $\alpha'= \alpha -1$.
    Otherwise, if $\alpha$ is a limit ordinal, then, by
    Def.~\ref{def:approximants}, we have
    $pm\SIM{\gamma}{\alpha''} qn\text{ for all }\alpha'' <\alpha$ and in
    particular $pm\SIM{\gamma}{\alpha'+1} qn$. So in either case we obtain  
    \begin{equation}
        \label{eq:game1}
        pm\SIM{\gamma}{\alpha'+1} qn\text{ for all }\gamma <\beta.
    \end{equation}
    If there is some $\omega$-transition that allows a response
    $qn\Step{a}{}{_\omega} q'n'$ that satisfies $p'm'\SIM{\beta'}{\alpha'}q'n'$,
    then \V\ picks this response and we can use the induction hypothesis
    to conclude that he wins the game from the next position.
    Otherwise, if no such $\omega$-transition exists,
    Equation~\eqref{eq:game1} implies that for every $\gamma <\beta$
    there is a response to some $q'n'$ that uses a non-$\omega$-transition $t(\gamma)$
    and that satisfies $p'm'\SIM{\gamma}{\alpha'}q'n'$.
    Since $\beta$ is a limit ordinal, there exist infinitely many $\gamma < \beta$.
    By the pigeonhole principle, that there must be one transition that occurs as
    $t(\gamma)$ for infinitely many $\gamma$. Therefore, a response that uses this transition
    satisfies $p'm'\SIM{\beta}{\alpha'}q'n'$.
    If \V\ uses this response, the game continues from position $(p'm',q'n',\alpha',\beta)$
    and he wins by induction hypothesis.

    For the ``if'' direction we show that $pm\notSIM{\beta}{\alpha}qn$ implies that \R\
    has a winning strategy in the approximation game from $(pm,qn,\alpha,\beta)$.
    In the base case of $\alpha=0$ or $\beta=0$ the implication holds
    trivially since the premise is false.
    By induction hypothesis we assume that the implication is true
    for all pairs pointwise smaller than $(\alpha,\beta)$.
    Observe that if $\alpha$ or $\beta$ are limit ordinals then
    (by Def.~\ref{def:approximants}) there are successors
    $\beta'\leq \beta$ and $\alpha'\leq \alpha$ s.t. $pm\notSIM{\beta'}{\alpha'}qn$. So without loss of
    generality we can assume that $\alpha$ and $\beta$ are successors.
    By the definition of approximants there must be a move $pm\step{a}p'm'$ s.t.
    \begin{itemize}
        \item for every possible response $qn\Step{a}{}{\omega}q'n'$ that uses
            some $\omega$-transition we have $p'm'\notSIM{\beta-1}{\alpha-1}q'n'$,
        \item for every possible response $qn\step{a} q'n'$ via
            some normal transition it holds that $p'm'\notSIM{\beta}{\alpha -1}q'n'$.
    \end{itemize}
    So if \R\ chooses $\alpha'=\alpha-1$, $\beta'=\beta-1$ and moves $pm\step{a}p'm'$
    then any possible response by \V\ will take the game to a
    position $(p'm',q'n',\gamma,\alpha')$ for a $\gamma\le \beta$.
    By induction hypothesis $\R$ wins the game.
    \qed
\end{proof}
\begin{lemma}
        \label{approximants:properties}
        For all ordinals $\alpha,\beta$ the following properties hold.
        \begin{enumerate}
            \item $pm\SIM{\beta}{\alpha}qn$ implies $pm'\SIM{\beta}{\alpha}qn'$ for all $m'\le m$ and $n'\ge n$
                \label{approximants:monotonicity}
            \item If $\alpha'\geq \alpha$ and $\beta'\ge \beta$ then
                $\SIM{\beta'}{\alpha'}\subseteq\SIM{\beta}{\alpha}$\label{approximants:alpha-inc}.
            \item There are ordinals $\CA,\CB$ such that 
                $\SIM{}{\CA} = \SIM{}{\CA+1}$ and $\SIM{\CB}{} = \SIM{\CB+1}{}$.
                \label{convergence_ordinals}
            \item $\SIM{}{}=\bigcap_{\alpha}\SIM{}{\alpha} = \bigcap_\beta\SIM{\beta}{}$\label{convergence}
        \end{enumerate}
    \end{lemma}

    The first point states that individual approximants are monotonic in the sense of Proposition
    \ref{reduction:monotonicity}.
    Points 2.-\ref{convergence}. imply that both $\SIM{}{\alpha}$ and $\SIM{\beta}{}$
    yield non-increasing sequences of approximants that converge towards simulation.
    It is easy to see that the approximants $\SIM{}{\alpha}$ do not
    converge at finite levels, and not even at $\omega$, i.e., $\CA>\omega$ in
    general. However, we will show that the approximants $\SIM{\beta}{}$ do
    converge at a finite level, i.e., $\CB\in\N$ for any pair of nets.

    \begin{proof}
     1) By Lemma~\ref{lem:game_interpretation} it suffices
     to observe that \V\ can reuse a winning strategy in the approximant game
     from $(pm,qn,\alpha,\beta)$ to win the game from $(pm-d_1,qn+d_2,\alpha,\beta)$
     for naturals $d_1, d_2$.

     2) If $pm\SIM{\beta'}{\alpha'}qn$ then, by Lemma~\ref{lem:game_interpretation}, \V\ wins
     the approximant game from position $(pm,qn, \beta', \alpha')$.
     By Lemma~\ref{approximants:game:monotonicity} he can also win the approximant game
     from $(pm,qn, \beta,\alpha)$. Thus $pm\SIM{\beta}{\alpha}qn$ by Lemma~\ref{lem:game_interpretation}.

     3) By point 2) we see that with increasing ordinal index $\alpha$ the approximant
     relations $\SIM{}{\alpha}$ form a decreasing sequence of relations,
     thus they stabilize for some ordinal $\CA$.
     The existence of a convergence ordinal for $\SIM{\CB}{}$ follows analogously.

     4) First we observe that
     $\bigcap_{\alpha}\SIM{}{\alpha} = \bigcap_{\alpha} \bigcap_\beta\SIM{\beta}{\alpha} =
     \bigcap_{\beta} \bigcap_\alpha \SIM{\beta}{\alpha}=
     \bigcap_\beta\SIM{\beta}{}$.
     It remains to show that $\SIM{}{} = \bigcap_{\alpha} \SIM{}{\alpha}$.

     To show $\SIM{}{} \supseteq \bigcap_{\alpha} \SIM{}{\alpha}$,
     we use $\CA$ from point \ref{convergence_ordinals}) and rewrite the right side to
     $\bigcap_{\alpha} \SIM{}{\alpha} = \SIM{}{\CA} = \SIM{}{\CA+1}$.
     From Definition \ref{def:approximants} we get that $\SIM{}{\alpha}=\SIM{\gamma}{\alpha}$
     for $\gamma\ge \alpha$ and therefore
     $\SIM{\CA+1}{\CA+1}=\SIM{}{\CA+1}=\SIM{}{\CA}=\SIM{\CA}{\CA}$.
     This means $\SIM{\CA}{\CA}=\bigcap_{\alpha}\SIM{}{\alpha}$ must be a simulation relation and hence a subset of $\msim$.

     To show $\SIM{}{} \subseteq \bigcap_{\alpha} \SIM{}{\alpha}$,
     we prove by ordinal induction that $\SIM{}{} \subseteq \SIM{}{\alpha}$ for all ordinals $\alpha$.
     The base case $\alpha=0$ is trivial.
     For the induction step we prove the equivalent property
     $\notSIM{}{\alpha} \subseteq \notSIM{}{}$.
     There are two cases.

     In the first case, $\alpha=\alpha'+1$ is a successor ordinal.
     If $pm\notSIM{}{\alpha'+1}qn$ then
     $pm\notSIM{\alpha'+1}{\alpha'+1}qn$ and therefore,
     by Lemma~\ref{lem:game_interpretation},
     \R\ wins the approximant game from
     $(pm,qn,\alpha'+1,\alpha'+1)$.
     Let $pm\step{a}p'm'$ be an optimal initial move by  \R.
     Now either there is no valid response and thus \R\ immediately wins
     in the simulation game or for every \V\ response $qn\step{a}q'n'$ we have
     $p'm'\notSIM{\alpha'}{\alpha'}q'n'$. Then also $p'm'\notSIM{}{\alpha'}q'n'$ and by induction
     hypothesis $p'm'\notSIM{}{}q'n'$.
     By Proposition~\ref{prop:game-char} we obtain that
     \R\ wins the simulation game from $(p'm',q'n')$ and thus from $(pm,qn)$.
     Therefore $pm\notSIM{}{}qn$, as required.

     In the second case, $\alpha$ is a limit ordinal. Then $pm\notSIM{}{\alpha}qn$ implies
     $pm\notSIM{}{\alpha'}qn$ for some $\alpha'< \alpha$
     and therefore $pm\notSIM{}{}qn$ by induction hypothesis.
     \end{proof}

The following lemma shows a certain uniformity property of the simulation
game. Beyond some fixed bound, an increased counter value of Spoiler can be
neutralized by an increased counter value of Duplicator, thus enabling Duplicator
to survive at least as many rounds in the game as before.

    \begin{lemma}\label{thm:approximants:colouring}
        For any one-counter net $N=(Q,\Alf,\delta)$ and $\omega$-net $N'=(Q'\Alf,\delta')$
        there is a fixed bound $c\in\N$ s.t. for all states $p\in Q,q\in Q'$, naturals $m'>m>c$
        and ordinals $\alpha$:
         \begin{equation}
             \label{thm1:eq}
             \forall n. ( pm\SIM{}{\alpha} qn \implies \exists n'. pm'\SIM{}{\alpha} qn')
         \end{equation}
    \end{lemma}
    \begin{proof}
        It suffices to show the existence of a local bound $c$ for any given
        pair of states $p,q$ that satisfies \eqref{thm1:eq}, since
        we can simply take the global $c$ to be the maximal such bound
        over all finitely many pairs $p,q$.
        Consider now a fixed pair $p,q$ of states.
        For $m,n\in\N$, we define the following (sequences of) ordinals.
        \begin{align*}
            I(m,n) = &\ \text{the largest ordinal $\alpha$ with } pm\SIM{}{\alpha} qn
            \text{ or }\CA\\ &\ \text{if no such $\alpha$ exists},\\
            I(m) = &\ \text{the increasing sequence of ordinals $I(m,n)_{n\ge 0}$},\\
            S(m) = &\ \sup\{I(m)\}.
        \end{align*}

        Observe that $I(m,n)$ can be presented as an infinite matrix where $I(m)$ is a column and
        $S(m)$ is the limit of the sequence of elements of column $I(m)$ looking upwards.
        Informally, $S(m)=lim_{i\rightarrow \infty} I(m,i).$

\begin{center}
 \includegraphics[angle=0]{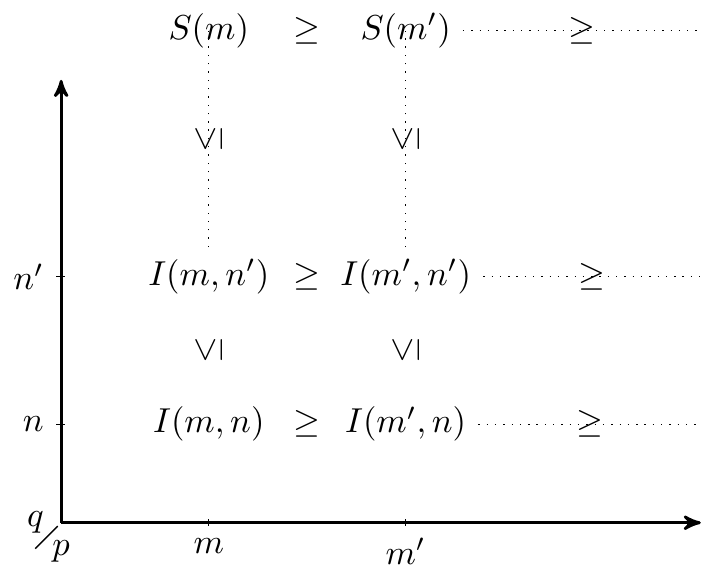}
\end{center}

        By Lemma~\ref{approximants:properties} (point \ref{approximants:monotonicity}), we derive
        that for any $n'>n\in\N$ and $m'>m\in\N$
        \begin{equation}
            I(m,n')\ge I(m,n) \ge I(m',n)\label{eq:approximants:thm1:mon1}
        \end{equation}
        and because of the second inequality also $S(m)\ge S(m')$.
        So the ordinal sequence $S(m)_{m\ge0}$ of suprema must be non-increasing and by the
        well-ordering of the ordinals there is a smallest index $k \in \N$ at
        which this sequence stabilizes:
        \[
            \forall l>k.\ S(l) = S(k).
        \]
        We split the remainder of this proof into three cases depending on whether $I(k)$ and $I(l)$
        for some $l>k$ have maximal
        elements. In each case we show the existence of a bound $c$ that satisfies requirement \eqref{thm1:eq}.

        Case 1. For all $l\ge k$ and $n\in\N$ it holds that $I(l,n)<S(l)$,
        i.e., no $I(l)$ has a maximal element.
        In this case $c := k$ satisfies the requirement \eqref{thm1:eq}.
        To see this, take $m'>m>c=k$ and $pm\SIM{}{\alpha} qn$.
        Then, by our assumption, $\alpha<S(m)$ and $S(m)=S(m')=S(k)$.
        Therefore $\alpha < S(m')$. Thus there must exist an $n' \in \N$ s.t.
        $pm'\SIM{}{\alpha} qn'$, as required.

        Case 2. For all $l\ge k$ there is a $n_l \in \N$ such that $I(l,n_l) = S(l)$,
        i.e., all $I(l)$ have maximal element $S(l)=S(k)$.
        Again $c := k$ satisfies the requirement \eqref{thm1:eq}.
        Given $m'>m>c=k$ and $pm\SIM{}{\alpha} qn$ we let $n' := n_{m'}$ and obtain
        $I(m',n') = S(m') = S(k) \ge\alpha$ and thus $pm'\SIM{}{\alpha} qn'$, as required.

        Case 3. If none of the two cases above holds then there must exist
        some $l > k$ s.t. the sequences $I(k), \dots, I(l-1)$ each have a maximal element and for
        $l' >l$ the sequence $I(l')$ has no maximal element.
        To see this, consider sequences $I(x)$ and $I(x')$ with $x' > x \ge k$.
        If $I(x')$ has a maximal element then so must $I(x)$, by equation
        \eqref{eq:approximants:thm1:mon1} and $S(x)=S(x')=S(k)$.
        Given this, we repeat the argument of the Case 1. with $c:=l$ and
        again satisfy the requirement \eqref{thm1:eq}. \qed
    \end{proof}

    \begin{lemma}\label{thm:approximants:convergence}
        Consider strong simulation $\SIM{}{}$ between a one-counter net
        $N=(Q,\Alf,\delta)$ and an $\omega$-net $N'=(Q',\Alf,\delta')$.
        There exists a constant $\CB\in\N$ s.t. $\SIM{}{}=\SIM{\CB}{}$.
    \end{lemma}
    \begin{proof}
        We assume the contrary and derive a contradiction. By
        Lemma~\ref{approximants:properties}(\ref{convergence}), 
        the inclusion
        $\SIM{}{} \subseteq \SIM{\beta}{}$ always holds for every ordinal $\beta$.
        Thus, if $\nexists \CB\in\N.\,\SIM{}{}=\SIM{\CB}{}$ then for every finite
        $\beta \in \N$ there are processes $p_0m_0$ and $q_0n_0$ s.t.
        $p_0m_0\SIM{\beta}{}q_0n_0$ but $p_0m_0 \notSIM{}{} q_0n_0$.
        In particular, this holds for the special case of $\beta = |Q\x Q'|(c +1)$,
        where $c$ is the constant given by
        Lemma~\ref{thm:approximants:colouring}, which we consider in the
        rest of this proof.

        Since $q_0n_0$ does not simulate $p_0m_0$, we can assume a winning strategy for \R\ in the
        simulation game which is optimal in the sense that it guarantees that
        the simulation level $\alpha_i$ -- the largest ordinal with $p_im_i\SIM{}{\alpha_i}q_in_i$ -- strictly decreases
        along rounds of any play.
        By monotonicity (Lemma~\ref{approximants:properties}, part~\ref{approximants:monotonicity})
        we can thus infer that whenever a pair of control-states repeats along a play, then \V's
        counter must have decreased or \R's counter must have increased:
        Along any partial play\\
        $(p_0m_0,q_0n_0)(t_0,t'_0)(p_1m_1,q_1n_1)(t_t,t'_1)\dots(p_km_k,q_kn_k)$ of length $k$ with
        $p_i=p_j$ and $q_i=q_j$ for some $i<j\le k$ we have $n_j < n_i$ or
        $m_j > m_i$.
        By a similar argument we can assume that \V\ also plays
        optimally, in the sense that he uses $\omega$-transitions to increase
        his counter to higher values than in previous situations with the same
        pair of control-states. By combining this with the previously stated
        property that the sequence of $\alpha_i$ strictly decreases we obtain
        the following:
        \begin{equation}
            \label{eq:approximants:thm2:loop}
            \text{if } p_i=p_j, q_i=q_j \text{ and }
            t'_{i-1},t'_{j-1}\in\delta_\omega \text{ then } m_j > m_i.
        \end{equation}
        Here $\delta_\omega$ denotes the set of transitions with symbolic effect $\omega$ in \V's
        net.

        Although \V\ loses the simulation game between $p_0m_0$ and $q_0n_0$, 
        our assumption $p_0m_0\SIM{\beta}{}q_0n_0$
        with $\beta = |Q\x Q'|(c +1)$ implies that \V\ does not lose with
        less than $\beta$ $\omega$-transitions, regardless of \R's strategy. Thus there
        always is a prefix of a play along which \V\ makes use of
        $\omega$-transitions $\beta$ times.\\
        Let $\pi = (p_0m_0,q_0n_0)(t_0,t'_0)(p_1m_1,q_1n_1)(t_t,t'_1)\dots(p_km_k,q_kn_k)$
        be such a partial play.

        Our choice of $\beta = |Q\x Q'|(c +1)$ guarantees that some pair
        $(p,q)$ of control-states repeats at least $c+1$
        times directly after \V\ making an $\omega$-step.
        Thus there are indices $i(1)<i(2)<\dots <i(c+1)<k$ s.t.
        for all $1\le j\le c+1$ we have $p_{i(j)} = p, q_{i(j)} = q$ and
        $t'_{i(j)}\in\delta_\omega$.
        By observation~\eqref{eq:approximants:thm2:loop} and $m_0 \ge 0$ we
        obtain that $m_{i(x)} \ge x$ for $0 \le x \le c+1$.
        In particular, $c \le m_{i(c)}<m_{i(c+1)}$, i.e., both of \R's counter
        values after the last two such repetitions must lie above $c$.
        This allows us to apply Lemma~\ref{thm:approximants:colouring} to derive a contradiction.

        Let $\alpha$ be the simulation level before this repetition: $\alpha$ is the largest ordinal
        with $pm_{i(c)}\SIM{}{\alpha}qn_{i(c)}$.
        Since $m_{i(c+1)} > m_{i(c)} > c$,
        Lemma~\ref{thm:approximants:colouring} ensures the existence of a natural $n'$ s.t.
        $pm_{i(c +1)}\SIM{}{\alpha}qn'$.
        Because \V\ used an $\omega$-transition in his last response leading to the repetition of
        states there must be a partial play $\pi'$ in which both players make the same moves as
        in $\pi$ except that \V\ chooses $n_{i(c +1)}$ to be $n'$.
        Now in this play we observe that the simulation level did in fact not strictly decrease as this
        last repetition of control-states shows: We have
        $pm_{i(c)}\SIM{}{\alpha}qn_{i(c)}\not\nsucceq_{\alpha+1} pm_{i(c)}$ and
        $pm_{i(c +1)}\SIM{}{\alpha}qn_{i(c +1)}$, which contradicts the
        optimality of \R's strategy. \qed
    \end{proof}

\section{The Main Result}

    We prove the decidability of simulation between one-counter nets and
    $\omega$-nets.
    First, we show that for each
    finite level $k\in\N$ the approximant $\SIM{k}{}$ is effectively
    semilinear, i.e., we can compute the semilinearity description of $\SIM{k}{}$.
    This yields a decision procedure for simulation that works as follows.
    Iteratively compute $\SIM{k}{}$ for growing $k$ and check after each round
    if the approximant has converged yet.
    The convergence test of $\SIM{k}{} \stackrel{?}{=} \SIM{k-1}{}$ can easily be done,
    since the approximants are semilinear sets.
    Termination of this procedure is guaranteed by
    Lemma~\ref{thm:approximants:convergence},
    and the limit is the simulation relation by
    Lemma~\ref{approximants:properties} (point \ref{convergence}).

    We recall the following important result by Jan\v{c}ar, Ku\v{c}era and Moller.
    \begin{theorem}[\cite{JKM2000}]\label{thm:main:reference}
        The largest strong simulation relation $\SIM{}{}$ between processes of two given one-counter
        nets is effectively semilinear.
    \end{theorem}

    Now we construct the semilinear approximants $\SIM{k}{}$.

    \begin{lemma}\label{lem:main:approximants}
    Given a one-counter net $N$ and an $\omega$-net $N'$,
    the approximant relations $\SIM{k}{}$ between them
    are effectively semilinear sets for all $k\in\N$.
    \end{lemma}
    \begin{proof}
        Let $N = (Q, \Act, \delta)$ and $N' = (Q',\Act,\delta')$.
        We prove the effective semilinearity of $\SIM{k}{}$ by induction on
        $k$.

        The base case $\SIM{0}{}=Q\x \N \x Q' \x \N$ is trivially effectively
        semilinear.

        For the induction step we proceed as follows. By induction hypothesis
        $\SIM{k}{}$ is effectively semilinear. Using this, we reduce the
        problem of checking $\SIM{k+1}{}$ between $N$ and $N'$ to the problem
        of checking normal strong simulation $\SIM{}{}$ between two derived
        one-counter nets $N_S$ and $N_D$, and obtain the effective
        semilinearity of the relation from Theorem~\ref{thm:main:reference}.
        More precisely, the derived one-counter nets $N_S$ and $N_D$ will
        contain all control-states of $N$ and $N'$, respectively,
        and we will have that $pm \SIM{k+1}{} qn$ w.r.t. $N,N'$ iff
        $pm \SIM{}{} qn$ w.r.t. $N_S,N_D$.

        Before we describe $N_S$ and $N_D$ formally, we explain the function
        of a certain test gadget used in the construction.

        An important observation is that after \V\ made an $\omega$-move in
        the approximation game between $N$ and $N'$, the winner of the game from
        the resulting configuration depends only
        on the control-states and \R's counter value, because \V\ could choose
        his counter arbitrarily high.
        Moreover, monotonicity (Lemma \ref{approximants:properties}, point
        \ref{approximants:monotonicity}) guarantees that there must be a minimal value
        for \R's counter with which he can win if at all. This yields the
        following property.

        For any pair of states $(p,q) \in Q\x Q'$ there must exist a value $M(p,q)\in\N\cup\{\omega\}$
        s.t. for all $m\in\N$
        \begin{equation}
            (\forall n \in\N.\, pm\notSIM{k}{} qn) \iff m\ge M(p,q) \label{claim:function_m}
        \end{equation}
        Since, by induction hypothesis, $\SIM{k}{}$ (and thus also its
        complement) is effectively semilinear, we can compute the values
        $M(p,q)$ for all $(p,q) \in Q\x Q'$.

        \paragraph*{The test gadgets.}
        Given the values $M(p,q)$, we construct test gadgets that check whether
        \R's counter value is $\ge M(p,q)$.
        For each $(p,q) \in Q\x Q'$ we construct two one-counter nets $S(p,q)$ and $T(p,q)$ with
        initial states $s(p,q)$ and $t(p,q)$, respectively, such that the
        following property holds for all $m,n\in\N$.
            \begin{equation}
                s(p,q)m\notSIM{}{} t(p,q)n \iff m\ge M(p,q) \label{eq:main:gadget}
            \end{equation}
        The construction of $S(p,q)$ and $T(p,q)$ is very simple.
        Let $s(p,q)$ be the starting point of a counter-decreasing
        chain of $e$-steps of length $M(p,q)\in\N$ where the last state of the chain
        can make an $f$-step whereas $t(p,q)$ is a simple $e$-loop (where $e,f$
        are fresh actions not in $\Act$).
        If $M(p,q)=\omega$, making $s(p,q)$ a deadlock suffices.
        Thus $S(p,q)$ and $T(p,q)$ are one-counter nets, denoted by
        $S(p,q) = (Q_s(p,q),\Act_s(p,q),\delta_s(p,q))$ and
        $T(p,q) = (Q_t(p,q),\Act_t(p,q),\delta_t(p,q))$.
        Wlog.\ we assume that their state sets are disjoint from each other and from the
        original nets $N,N'$.

        \paragraph*{The construction of $N_S$ and $N_D$.}
        Let $N_S=(Q_S,\Act',\delta_S)$ and $N_D=(Q_D,\Act',\delta_D)$ be
        one-counter nets constructed as follows.
        $\Act' = \Act\cup Q\x Q'\cup \{f,e\}$ (where $e,f$ are the actions
        from the test gadgets).
        \R's new net $N_S$ has states
        \begin{equation}
            Q_S = Q \cup \bigcup_{p\in Q, q\in Q'} Q_s(p,q)
        \end{equation}
        \V's new net $N_D$ has states
        \begin{equation}
            Q_D = Q' \cup \{W\}\cup\bigcup_{p\in Q, q\in Q'} Q_t(p,q).
        \end{equation}
        where $W$ is a new state.

        Now we define the transition relations.
        $\delta_S = \delta \cup \bigcup_{p\in Q, q\in Q'} \delta_s(p,q)$
        plus the following transitions for all
        $p\in Q, q\in Q'$:
        \begin{equation}
            p\step{(p,q),0}s(p,q) \label{eq:spoilmove}
        \end{equation}

        $\delta_D = \{q\step{a,x}q' \in \delta'\ |\ x \neq \omega\} \cup \bigcup_{p\in Q, q\in Q'} \delta_t(p,q)$
        plus the following transitions
        for all $p,p'\in Q$ and $q,q'\in Q'$:
        \begin{eqnarray}
            q\step{a,0}t(p,q')         &,&\mbox{ if } q\step{a,\omega}q' \in \delta' \label{eq:omegamove} \\
            q\step{(p',q'),0} W       & & \label{eq:punish_wrong_spoilmove}\\
            t(p,q)\step{(p,q),0}t(p,q) & & \label{eq:stay}\\
            t(p,q)\step{(p,q'),0}W     & &\mbox{ for all } q\neq q' \label{eq:punish1}\\
            t(p,q)\step{a,0}W          & &\mbox{ for all } a \in \Act \label{eq:punish2}\\
            W\step{a,0}W                 & &\mbox{ for all } a \in \Act'
        \end{eqnarray}

%

        \paragraph*{Correctness proof.}
        We show that for any pair $pm, qn$ of configurations of the nets $N,N'$ we have
        $pm\SIM{k+1}{}qn$ if and only if $pm\msim qn$ in the newly constructed nets $N_S,N_D$.

        To prove the `if' direction we assume that $pm\notSIM{k+1}{}qn$
        w.r.t. $N,N'$ and derive that $pm\notSIM{}{} qn$ w.r.t. $N_S,N_D$.
        By our assumption and Definition~\ref{def:approximants}, there exists
        some ordinal $\alpha$ s.t. $pm\notSIM{k+1}{\alpha}qn$.
        By Lemma~\ref{lem:game_interpretation},
        \R\ has a winning strategy in the approximation game
        from position $(pm, qn,\alpha, k+1)$. The result then follows from the
        following general property.

        \paragraph*{Property P1.}
        For all ordinals $\alpha$,
        control-states $(p,q) \in Q \x Q'$ and naturals $m,n \in \N$:
        If \R\ has a winning strategy in the approximation game
        from position $(pm, qn, \alpha, k+1)$ then he also
        has a winning strategy in the strong simulation game between $N_S,N_D$ from
        position $(pm, qn)$.

        \begin{proof}
        To prove P1, we fix some $p\in Q, q\in Q'$ and proceed by ordinal induction on $\alpha$.
        The base case trivially holds since \R\ looses from a position $(pm,qn,0,k+1)$.

        For the induction step let \R\ play the same move $pm \step{a} p'm'$
        for some $a \in \Act$ in both games according to his assumed winning strategy in the
        approximation game.
        Now \V\ makes his response move in the new game between $N_S,N_D$, which yields two cases.
        In the first case, \V\ does not use a transition from
        Equation~\eqref{eq:omegamove}. Then
        his move induces a corresponding move in the approximation game which
        leads to a new configuration $(p'm', q'n', \alpha', k+1)$ where
        $p'm'\notSIM{k+1}{\alpha'}q'n'$ for some ordinal $\alpha'<\alpha$.
        Thus, by Lemma~\ref{lem:game_interpretation} and the induction
        hypothesis, the property holds.

        In the second case, \V's response is via a transition from
        Equation~\eqref{eq:omegamove}, which leads to a new configuration
        $(p'm', t(p'',q')n)$ for some $p''\in Q$.
        Thus in the approximants game there will exist \V\ moves to
        positions $(p'm', q'n', \alpha', k)$ where $n' \in \N$ can be arbitrarily high.
        We can safely assume that \V\ chooses $p''=p'$, since otherwise \R\ can win in one round by
        playing $p'm'\step{(p',q')}$.
        Now in the following round \R\ can play $p'm'\step{(p',q')}s(p',q')m'$
        by Equation~\eqref{eq:spoilmove} and \V's only option
        is to stay in his current state by Equation~\eqref{eq:stay}.
        The game thus continues from
        $(s(p',q')m', t(p',q')n)$.
        By our assumption \R\ wins the approximation game from the position $(pm,qn, \alpha, k+1)$.
        Thus there is some ordinal $\alpha'<\alpha$ s.t.
        \R\ also wins the approximation game from the position $(p'm', q'n', \alpha', k)$
        for every $n'\in \N$.
        Thus, by Lemma~\ref{lem:game_interpretation} and Definition~\ref{def:approximants},
        we have $p'm' \notSIM{k}{\alpha'} q'n'$ and by
        Lemma \ref{approximants:properties} (item~\ref{approximants:alpha-inc})
        $p'm' \notSIM{k}{} q'n'$ for all $n' \in \N$.
        By Equation \eqref{claim:function_m} we obtain $m'\geq M(p',q')$.
        By the construction of the gadgets and Equation \eqref{eq:main:gadget}
        we get $s(p',q')m'\notSIM{}{} t(p',q')n$, which implies the desired
        property.
    \end{proof}
        This concludes the proof of the `if' direction.
        Now we prove the `only if' direction of the correctness property.
        We assume that $pm\notSIM{}{} qn$ in the newly constructed nets $N_S$
        and $N_D$ and derive that
        $pm\notSIM{k+1}{}qn$ w.r.t. $N,N'$.
        To do this, we first show the following general property.
        \paragraph*{Property P2.}
        If $pm \notSIM{}{} qn$ with respect to nets $N_S$ and $N_D$ then
        there exists some general ordinal $\alpha'$ s.t. $pm \notSIM{k+1}{\alpha'} qn$ with respect to
        nets $N,N'.$
    \begin{proof}
	Assume $pm \notSIM{}{} qn$ with respect to nets $N_S$ and $N_D$.
        Since both $N_S,N_D$ are just one-counter nets,
        non-simulation manifests itself at some finite approximant
        $\alpha \in \N$, i.e., $pm\notSIM{}{\alpha} qn$.
        We prove property P2 by induction on $\alpha$. The base case of
        $\alpha=0$ is trivial. For the induction step we consider a
        move $pm \step{a} p'm'$ for some $a \in \Act$ by \R\ in both games
        according to \R's assumed winning strategy in the game between $N_S,N_D$
        (It cannot be a \R\ move
        $p\step{(p,q),0}s(p,q)$ by Equation~\eqref{eq:spoilmove}, because
        \V\ would immediately win via a reply move by
        Equation~\eqref{eq:punish_wrong_spoilmove}).
        Now we consider all (possibly infinitely many) replies by \V\ in the approximation
        game between $N,N'$ from a position $(pm,qn,\alpha', k+1)$ for some yet to be
        determined ordinal $\alpha'$.
        These replies fall into two classes.

        In the first class, \V's move $qn \step{a} q'n'$ is not due to
        an $\omega$-transition and thus also a possible move in the
        strong simulation game between $N_S,N_D$.
        From our assumption that \R\ wins the strong
        simulation game from position $(pm, qn)$ in at
        most $\alpha$ steps, it follows that \R\ wins the strong simulation
        game from $(p'm',q'n')$ in at most $\alpha-1$ steps.
        By induction hypothesis, there exists an ordinal $\alpha''$
        s.t. \R\ has a winning strategy in the
        approximation game for $\SIM{k+1}{\alpha''}$ between $N,N'$
        from position $(p'm', q'n')$.
        There are only finitely many such replies.
        Thus let $\alpha^0$ be the maximal such $\alpha''$.

        In the second class, \V's move $qn \step{a} q'n'$ uses an
        $\omega$-transition which does not exist in $N_D$.
        Instead there exists a \V\ transition
        $qn \step{a,0}t(p'',q')n$ by Equation~\eqref{eq:omegamove}.
        From our assumption that \R\ wins the strong
        simulation game from position $(pm, qn)$ in at
        most $\alpha$ steps, it follows that \R\ wins the strong simulation
        game from $(p'm',t(p'',q')n)$ in at most $\alpha-1$ steps.
        If $p'' \neq p'$ then this is trivially true by a \R\ move by
        Equation~\eqref{eq:spoilmove}. Otherwise, if $p''=p'$,
        then this can only be achieved by a \R\ move of
        $p'm' \step{(p',q'),0} s(p',q')m'$ in the next round,
        because for any other \R\ move \V\ has a winning countermove
        by Equations~\eqref{eq:punish1} or \eqref{eq:punish2}.
        In this case \V\ can only reply with a move
        $t(p',q')qn \step{(p',q'),0}t(p',q')n$ by Equation~\eqref{eq:stay},
        and we must have that \R\ can win in at most $\alpha-2$ steps
        from position $(s(p',q')m', t(p',q')n)$.
        This implies, by Equation~\eqref{eq:main:gadget}, that
        $m' \ge M(p',q')$.
        Then Equation~\eqref{claim:function_m} yields
        $\forall n \in\N.\, p'm'\notSIM{k}{} q'n$.
        Thus for every $n \in \N$ there exists some ordinal
        $\alpha_n$ s.t. $p'm'\notSIM{k}{\alpha_n} q'n$.
        Let $\alpha''$ be the smallest ordinal s.t.
        $\forall n \in \N.\, \alpha_n \le \alpha''$.
        Each of the finitely many distinct $\omega$-transitions yields such
        an $\alpha''$. Let $\alpha^1$ be the maximum of them.

        We set $\alpha' := \max(\alpha^0,\alpha^1)+1$.
        Then every reply to \R s move $pm\step{a}p'm'$ in the approximation game
        from $(pm,qn,\alpha',k+1)$ leads to some position that is winning for \R.
        So, \R\ has a winning strategy in the
        approximation game from $(pm,qn,\alpha',k+1)$ and
        by Lemma~\ref{lem:game_interpretation},
        $pm\SIM{k+1}{\alpha'}qn$ w.r.t. $N,N'$, which concludes the proof
        of property P2.
    \end{proof}
        To show the `only if' direction of the correctness property, we assume that
        $pm\notSIM{}{} qn$ in the newly constructed nets $N_S,N_D$.
        By property P2 we have $pm \notSIM{k+1}{\alpha'} qn$ for some ordinal $\alpha'$ and thus
        $pm\notSIM{k+1}{}qn$ w.r.t. $N,N'$. This concludes the `only if' direction.

        We have constructed one-counter nets $N_S,N_D$ s.t.
        $pm\SIM{k+1}{}qn$ w.r.t. $N,N'$ if and only if $pm\msim qn$
        w.r.t. $N_S,N_D$.
        By Theorem~\ref{thm:main:reference}, $\SIM{k+1}{}$ is
        effectively semilinear. \qed
    \end{proof}

    \begin{theorem}\label{thm:main}
        The largest weak simulation over processes of a given one-counter net is
        effectively semilinear and thus decidable.
    \end{theorem}
    \begin{proof}
        By Theorem~\ref{thm:reduction} it suffices to show that the largest strong simulation
        relation $\SIM{}{}$ between a one-counter net $N$ and a $\omega$-net $N'$ is effectively
        semilinear.
        By Lemma~\ref{lem:main:approximants}, we can iteratively compute the semilinearity description of
        the approximants $\SIM{k}{}$ for $k=0,1,2,\dots$.
        Convergence can be detected by checking
        if $\SIM{k}{} \stackrel{?}{=} \SIM{k-1}{}$, which is effective
        because equality is decidable for semilinear sets.
        Termination (i.e., eventual convergence at a finite index) of this procedure is guaranteed by
        Lemma~\ref{thm:approximants:convergence},
        and the reached limit is the semilinear simulation relation by
        Lemma~\ref{approximants:properties} (item \ref{convergence}). \qed
    \end{proof}

\section{Approximant Convergence at $\omega^2$}\label{sec:omegasquare}

We show that ordinary weak simulation approximants $\WSIM{}{\alpha}$ converge at level
$\alpha=\omega^2$ on \OCN.

\begin{lemma}\label{lem:omegasquare}
When considering relations between a one-counter net and an $\omega$-net,
we have $\SIM{}{\omega i}\subseteq \SIM{i}{}$ for every $i \in \N$.
\end{lemma}
\begin{proof}
    By induction on $i$.
    The base case of $i=0$ is trivial, since $\SIM{0}{}$ is the full
    relation.
    We prove the inductive step by assuming the contrary and deriving a
    contradiction.
    Let $pm\SIM{}{\omega i} qn$ and
    $pm \notSIM{i}{} qn$ for some $i>0$.
    Then there exists some ordinal $\alpha$ s.t. $pm \notSIM{i}{\alpha} qn$.
    Without restriction let $\alpha$ be the least ordinal
    satisfying this condition.
    If $\alpha \le \omega i$ then we trivially have a
    contradiction. Now we consider the case $\alpha > \omega i$.
    By $pm \notSIM{i}{\alpha} qn$ and Lemma~\ref{lem:game_interpretation},
    \R\ has a winning strategy in the approximant game from position
    $(pm,qn,\alpha,i)$. Without restriction we assume that \R\ plays optimally,
    i.e., wins as quickly as possible. Thus this game must reach some game position
    $(p'm',q'n',\alpha'+1,i)$ where $\alpha' \ge \omega i$ is a limit ordinal,
    such that  \R\ can win from $(p'm',q'n',\alpha'+1,i)$ but not from
    $(p'm',q'n',\alpha',i)$.
    I.e., $p'm' \notSIM{i}{\alpha'+1} q'n'$, but $p'm' \SIM{i}{\alpha'} q'n'$.
    Consider \R's move $p'm'\step{a} p''m''$ according to his optimal winning
    strategy in the game from position $(p'm',q'n',\alpha'+1,i)$.
    Since $p'm' \SIM{i}{\alpha'} q'n'$ and $\alpha'$ is a limit ordinal,
    for every ordinal
    $\gamma_k < \alpha'$, \V\ must have some countermove
    $q'n'\step{a} q_kn_k$ s.t. $p''m'' \SIM{j}{\gamma_k} q_kn_k$,
    where $j=i-1$ if the move was due to an $\omega$-transition and $j=i$
    otherwise. In particular, $\sup_k \{\gamma_k\} = \alpha'$.
    However, since \R's move $p'm'\step{a} p''m''$ was according to his optimal winning
    strategy from position $(p'm',q'n',\alpha'+1,i)$, we have
    that $p''m'' \notSIM{j}{\alpha'} q_kn_k$.
    Therefore, there must be infinitely many different \V\ countermoves
    $q'n'\step{a} q_kn_k$. Infinitely many of these countermoves must be due to
    an $\omega$-transition, because apart from these the system is finitely
    branching.
    Thus for every ordinal $\gamma < \alpha'$ there is some
    \V\ countermove $q'n'\step{a} q_kn_k$ which is due to an $\omega$-transition
    s.t. $p''m'' \SIM{i-1}{\gamma_k} q_kn_k$ where $\gamma_k \ge \gamma$ (note the
    $i-1$ index due to the $\omega$-transition).
    In particular, we can choose $\gamma=\omega (i-1)$, because $i>0$ and $\alpha'
    \ge \omega i$.
    Then we have $p''m'' \SIM{i-1}{\omega (i-1)} q_kn_k$, but
    $p''m'' \notSIM{i-1}{\alpha'} q_kn_k$.
    However, from $p''m'' \SIM{i-1}{\omega (i-1)} q_kn_k$ and the induction hypothesis,
    we obtain $p''m'' \SIM{i-1}{} q_kn_k$
    and in particular
    $p''m'' \SIM{i-1}{\alpha'} q_kn_k$. Contradiction.
\end{proof}

\begin{theorem}
Weak simulation approximants on \OCN\ converge at level $\omega^2$,
but not earlier in general.
\end{theorem}
\begin{proof}
First we show that $\WSIM{}{\omega^2}$ is contained in $\WSIM{}{}$ for \OCN.
Let $pm$ and $qn$ be processes of \OCN\ $M$ and $N$, respectively.
Let $M',N'$ be the derived \OCN\ and $\omega$-net from
Theorem~\ref{thm:reduction}.
Assume $pm \WSIM{}{\omega^2} qn$ w.r.t. $M,N$. 
Then, by point \ref{thm:reduction:ordinals}) of Theorem~\ref{thm:reduction}, 
$pm \SIM{}{\omega^2} qn$ w.r.t. $M',N'$.
In particular we have $pm \SIM{}{\omega \CB} qn$ w.r.t. $M',N'$,
with the $\CB \in \N$ from Lemma~\ref{thm:approximants:convergence}.
From Lemma~\ref{lem:omegasquare} we obtain
$pm \SIM{\CB}{} qn$ w.r.t. $M',N'$.
Lemma~\ref{thm:approximants:convergence} yields 
$pm \SIM{}{} qn$ w.r.t. $M',N'$.
Finally, by Theorem~\ref{thm:reduction}, we obtain
$pm \WSIM{}{} qn$ w.r.t. $M,N$.  

To see that $\omega^2$ is needed in general, consider the following class of
examples. Let $p \step{a} p$ define a simple \OCN\ (actually even a finite
system). For every $i \in \N$ we define an \OCN\ $N_i$ with
transitions $(q_k, a, -1, q_k)$, 
$(q_{k-1}, \tau, 0, q_{k-1}')$,
$(q_{k-1}', \tau, 1, q_{k-1}')$, and
$(q_{k-1}', a, 0, q_{k})$
for all $k$ with $1 \le k \le i$.
Then, for the net $N_i$, we have $p \WSIM{}{\omega i} q_00$,
but $p \notWSIM{}{} q_00$.
Thus in general $\WSIM{}{} \neq \WSIM{}{\omega i}$ for any $i \in \N$.
\end{proof}

\section{Undecidability of Trace Inclusion and Equivalence}\label{sec:traces}
For any process $\alpha$ we write $T(\alpha)$ for the set $\{w\in\Act^* | \exists \beta.\alpha\step{w}\beta\}$
of \emph{traces} of $\alpha$.
We consider trace inclusion and equivalence checking for \OCN, which was
stated as an open question in \cite{EN94}.
We show that both problems are undecidable for \OCN\ by a reduction from the containment problem
for weighted automata \cite{ABK2011}.

\begin{definition}[Weighted Automata]
    A \emph{weighted finite automaton (WFA)} 
    is a tuple
    $(Q,\Sigma,\delta,q_0)$ where $Q$ is a finite set of states,
    $\Sigma$ a finite alphabet, $q_0\in Q$ an initial state and $\delta\subseteq Q\x\Sigma\x\N\x Q$
    a transition relation.
    If $(p,a,d,p')\in\delta$ the automaton can go from $p$ to $p'$ reading a symbol `$a$' with reward $d\in \N$.
    A \emph{run} of A on a word $w=w_0w_1\dots w_n\in\Sigma^*$ is a sequence
    $(q_i,w_i,d_i,q_{i+1})_{0\le i<n} \in \delta^*$ of transitions.
    The \emph{value} of such a run is $\sum_{i=0}^{n-1}d_i$; the value $L(A,w)$ of a word $w\in\Sigma^*$
    is the maximal value of any run on $w$.
\end{definition}

We say that the language of WFA $A$ is contained in that of WFA $B$ (over the same alphabet $\Sigma$),
$L(A)\subseteq L(B)$, iff for all words $w\in\Sigma^*$, $L(A,w)\le L(B,w)$.
Checking $L(A)\subseteq L(B)$ is undecidable \cite{ABK2011} (Theorem 4).
The next result is a direct consequence.

\begin{theorem}
Trace inclusion/equivalence between \OCN\ processes is undecidable.
\end{theorem}
\begin{proof}
    Inclusion can trivially be reduced to equivalence for nondeterministic
    systems like \OCN.  
    Thus we show undecidability of inclusion by reduction from WFA
    containment. 
    The idea is to encode the WFA
    as \OCN, using the counter as accumulator. To ensure a faithful encoding
    of WFA containment, the \OCN\ can at any point jump to a gadget that
    compares the counter values.

    Given WFA $A=(Q_A,\Sigma,\delta_A,q_A)$ and $B=(Q_B,\Sigma,\delta_B,q_B)$
    we construct nets $A'$ and $B'$ with states $Q_A\cup\{D\}$ and $Q_B\cup\{D\}$
    resp., over alphabet $\Act=\Sigma\cup\{d\}$ where $d$ is a fresh symbol.
    We add transitions $D\step{d,-1}D$ to both nets as well as $q\step{d,-1}D$
    for any original state.
    We argue that $L(A)\subseteq L(B)$ if and only if $T(q_A(0))\subseteq T(q_B(0))$.

    Assume a witness $w$ with $L(A,w)=v>L(B,w)$. Then there is a run of $A$ on $w$
    with a value higher than that of any run of $B$ on $w$.
    So the word $w d^v$ must be a valid trace from $q_A(0)$, but not from $q_B(0)$.
    Conversely, if $L(A,w)\le L(B,w)$ for all $w\in\Sigma^*$, then
    for any run of $A'$ there is a run of $B'$ over the same sequence of actions
    which accumulates a higher or equal counter value. Thus no such word can be extended to
    a counterexample for trace inclusion by appending finitely many $d$'s.
\end{proof}

\section{Comparing \OCN/\OCA\ and Finite Systems}\label{sec:finite}

{\it\noindent Simulation.}
First we consider checking strong/weak simulation between 
\OCN/\OCA\ and finite systems, and vice-versa.

\begin{theorem}\label{thm:fss_ws_OCN_in_P}
     Checking if a finite-state process weakly simulates a \OCN\ process is in $P$.
\end{theorem}
\begin{proof}
    It suffices to first replace the step relation in the finite system with
    its weak closure so that $q\wstep{a}q'\iff q\step{a}q'$ and then check if the resulting
    finite system strongly simulates the net.
    The finiteness of the state space allows us to compute the weak closure in polynomial
    time. A polynomial time algorithm for checking strong simulation between \OCN\ and
    finite-state processes can be found in \cite{Kuc2000a}.
\end{proof}

For the other direction, checking if a \OCN\  process weakly simulates a finite-state process,
we show that it suffices to consider a finite version of the net where the counter
is capped at a polynomially bounded level.
The crucial observation is that monotonicity implies that \V\ must be able to
ensure that his counter never decreases along any partial play that repeats
control-states.

\begin{definition}
    Let $N=(Q,\Act,\delta)$ be a \OCN\  and $l\in\N$.
    The \emph{$l$-capped version} of $N$ is the finite system $N_l=(Q_l, \step{})$ with states
    $Q_l=\{(\hat{q},n)| q\in Q, n\le l\}$ and transitions
    $(\hat{q},n)\step{a}(\hat{q}',\min\{n',l\})$ iff $qn\Step{a}{}{N}q'n'$.
\end{definition}
It is easy to see that $N_l$ can be constructed from $N$ in time proportional to $l \times |N|$.
For $n,l \in \N$ we observe the following properties.
\begin{proposition}\
    \label{prop:fs:capped}
    \begin{enumerate}
        \item $(\hat{q},\min\{n,l\}) \mwsim qn$,\label{obs:fs:sim}
        \item $qn \WSIM{}{l} (\hat{q},\min\{n,l\})$,\label{obs:fs:sim-up-to-l}
        \item $(\hat{q},\min\{n,l\})\mwsim (\hat{q},\min\{n+1,l\})$.\label{obs:fs:monotonicity}
    \end{enumerate}
\end{proposition}

We continue to show that \V\ can be assumed to play optimally in a sense
that depends on cycles in the underlying control graphs.
Consider a simulation game between a finite process and a \OCN\  process (or its $c$-capped version).

\begin{definition}
    A partial play $\pi=(p_0,q_0n_0)(t_0,t_0')\dots$ $(p_l,q_ln_l)$
    is a \emph{cycle} if $p_0=p_l$ and $q_0=q_l$.
    It is \emph{decreasing} if $n_0 > n_l$.

    Similarly, if \V\ plays on a $c$-capped version of the net, the play
    $\pi=$ $(p_0,(\hat{q}_0,\min\{n_0,c\}))$ $(t_0,t_0')\dots$ $(p_l,(\hat{q}_l,\min\{n_l,c\}))$ is a cycle if
    $p_0=p_l$ and $\hat{q}_0=\hat{q}_i$ and is decreasing if the second component of \V's state
    decreases.
    A cycle is called \emph{simple} if no proper subsequence is itself a cycle.
    The length of simple cycles is bounded by $|S\x Q|$, where $S$ is the set
    of states of the finite-state process.
\end{definition}

\begin{lemma}
    \label{lem:fs:optimal}
     Suppose $p\mwsim qn$. Then \V\ has a winning strategy in the weak simulation game
     that moreover guarantees the following properties in every play.
     \begin{enumerate}
         \item No round decreases the counter by more than $|Q|*2+1$.
         \item Every simple cycle is non-decreasing.
     \end{enumerate}
\end{lemma}
   \begin{proof}
    A weak step $s_0(m_0)\wstep{a}t_l(n_l)$ by \V\ is due to some sequence
    $s_0(m_0)\step{\tau}s_1(m_1)$ $\step{\tau}\dots$
    $\step{\tau}s_l(m_l)\step{a}t_0(n_0)\step{\tau}t_1(n_1)\step{\tau}\dots\step{\tau}t_k(n_k)$.
    By monotonicity it is suboptimal for \V\ to decrease the counter
    when silently moving from state $s_i$ to $s_j=s_i$ (or from $t_i$ to $t_j=t_i$) for $i<j$.
    Also, we can safely assume that a weak step as above will be non-decreasing if
    there are indices $i<j$ with $s_i=s_j$ and $m_i<m_j$ (or $t_i=t_j$ and $n_i<n_j$).
    Therefore, if the weak step decreases the counter, both silent paths will be acyclic and hence no longer than $|Q|$.
    Such a step cannot decrease the counter by more than $|Q|*2+1$.

    For the second point observe that if \V\ cannot avoid the next simple cycle to be
    decreasing, then \R\ must have some strategy to enforce cycles to be decreasing.
    Such a strategy must be winning for \R\ as it eventually exhausts \V's counter.
\end{proof}

The next lemma uses the previously stated optimality assumption to show that
we only need to consider a polynomially capped net to determine if
a \OCN\  process weakly simulates a finite-state process.

\begin{lemma}\label{lem:fs:polycap}
    For any pair $F=(S,\step{})$, $N=(Q,\Act, \delta)$ of a finite-state system and \OCN\  resp.,
    there is a fixed polynomial bound $c$ such that for all $n\in\N$:
    \[
        p\mwsim qn \iff p\mwsim (\hat{q},\min\{n,c\})
    \]
    and $(\hat{q},\min\{n,c\})$ is a state of the $c$-capped version $N_c$ of $N$.
\end{lemma}
  \begin{proof}
    The "if" direction follows directly from Proposition \ref{prop:fs:capped} (point \ref{obs:fs:sim}).
    For the other direction we show that $c= (2|Q|+1)(|S\x Q|+1)$ suffices to
    contradict $p\mwsim qn$ and
    $p\not\mwsim(\hat{q},\min\{n,c\})$.

    If $p\mwsim qn$ then $p\WSIM{}{c} qn$ and by Proposition~\ref{prop:fs:capped} (\ref{obs:fs:sim-up-to-l})
    we have $p\WSIM{}{c} (\hat{q},\min\{n,c\})$.
    Moreover, \V\ has an \emph{optimal} strategy in the sense of Lemma~\ref{lem:fs:optimal}.
    We see that using the same strategy in the game $p$ vs.\ $(\hat{q},\min\{n,c\})$ guarantees that
    \begin{enumerate}
      \item No round decreases the second component of \V's state by more than $|Q|*2+1$.
      \item For any simple cycle between game positions $p_i,(\hat{q}_i,n_i)$ and $p_j,(\hat{q}_j,n_j)$ it holds that $n_j\ge n_i$
          or $n_j\ge c-(|Q|*2+1) * (|S\x Q|)$.
    \end{enumerate}
    To see the second point observe that the only way a simple cycle can be decreasing is because some of
    its increases are dropped due to the counter being at its limit $c$.
    Then point 1 implies $n_j\ge c-(2|Q|+1)(|S\x Q|)$ because the length of simple cycles is bounded by
    $|S\x Q|$.

    By our assumption $p\not\mwsim (\hat{q},\min\{n,c\})$, we can consider a play
    \[
        \pi=(p_0,(\hat{q}_0,n_0))(t_0,t_0') (p_1,(\hat{q}_1,n_1))(t_1,t_1')\dots(p_l,(\hat{q}_l,n_l))
    \]
    where $p_0=p$ and $(\hat{q}_0,n_0)=(\hat{q},\min\{n,c\})$,
    along which \V\ plays optimally as described above and which is winning for \R\ in the smallest
    possible number of rounds.

    Since $c>|S\x Q|$, we know that $\pi$ must contain cycles as otherwise $l\le |S \x Q|$ and thus
    $p\WSIM{}{c} (q,\min\{n,c\})$ contradicts that $\pi$ is won by \R.
    So assume the last simple cycle in $\pi$ is between positions $i$ and $j$.
    We know that
    $n_j < n_i$, as otherwise omitting this last cycle results in a shorter winning play
    for \R\ by monotonicity.

    Thus, by Observation 2 above, $n_j\ge c-(2|Q|+1)(|S\x Q|)$
    and in particular
    we get that $n_j\ge(2|Q|+1)$
    due to our choice of $c$.

    Finally, recall that the last
    position
    $p_l,(\hat{q}_l,n_l)$
    in the play $\pi$ must be directly winning
    for \R.
    That is, for some action $a$ it holds that $p_l\step{a}$ and
    $(\hat{q_l},n_l)\notwstep{a}$.
    But since $n_l>(2|Q|+1)$, we know that also $q_ln_l\notwstep{a}$ in the original \OCN\ process.
    This contradicts our assumption that \V's original strategy in the unrestricted game
    was winning.
\end{proof}

\begin{theorem}
     Checking if a \OCN\  process weakly simulates a finite-state process can be done in polynomial
     time.
\end{theorem}

\begin{proof}
    To check if $p\mwsim qn$ holds we can by Lemma \ref{lem:fs:polycap} equivalently check 
    $p\mwsim (\hat{q},\min\{n,c\})$
    where $(\hat{q},\min\{n,c\})$ is a state of a polynomially bounded finite system $N_c$.
    Checking weak simulation between two finite processes is in $P$.
\end{proof}

{\bigskip\noindent\it Trace Inclusion.}
Now we consider checking strong/weak trace inclusion between finite-state systems
and \OCA/\OCN.
It is undecidable whether an \OCA\ contains the strong/weak traces of a
finite-state system \cite{Valiant1973}.
However, is is decidable whether a Petri net contains the strong/weak traces of a
finite-state system \cite{JEM:JCSS1999}, and thus the question is decidable
for \OCN.

\newcommand{\ppath}[1]{\Step{+}{}{#1}}
\newcommand{\path}[1]{\Step{}{}{#1}}
\newcommand{\reach}{\Step{}{*}{}}

Now we consider the other direction of trace inclusion.
We show that checking whether a finite-state system contains the strong/weak
traces of an \OCA\ is PSPACE complete. 
For this we recall some structural properties of \OCA\ processes.
We write $pm\path{k}qn$ for \OCA\ configurations $pm$ and $qn$ if there
is a path of length $k$ from $pm$ to $qn$.
A path is \emph{positive} if at most the last visited configuration has
counter value $0$,
i.e., no step is due to a transition in $\delta_0$.
We write $pm\ppath{k}qn$ if there is a positive path of length $k$ from $pm$ to $qn$.

\begin{lemma}[\cite{EWY2008}, Lemma 5]\label{1-0-reach}
    Consider a \OCA\ with stateset $Q$ and let $p,q\in Q$.
    If $p1\reach q0$ then $p1\path{k}q0$ for some $k\le|Q|^3$.
\end{lemma}
\begin{lemma}\label{lem:bound}
    Consider a \OCA\ $(Q,Act,\delta,\delta_0)$ where $K=|Q|$ and $p,q\in Q$.
    If $pm\reach qn$ for some $n\in\N$ then $pm\path{k}qn'$ for some $n'$ and
    $k \le \max\{m,1\}5K^4$.
\end{lemma}

\begin{proof}
    We distinguish two cases depending on whether there is a positive minimal path from $pm$
    to $qn$.

    Case 1: There is a positive minimal path witnessing $pm\reach qn$.
    Consider such a path $(p_0m_0), (p_1m_1),\dots,(p_km_k)$ from $pm=p_0m_0$ to $qn=p_km_k$.
    We know that there is a path from $p$ to $q$ in the control graph of the
    automaton that uses transitions in $\delta$ only. So there must be such a path
    in the control graph that is no longer than $K$. 
    Thus, if $m\ge K$, then there is a $n'$ such that 
    $pm\ppath{k}qn'$ for some $k \le K$.
    Otherwise, if $m<K$, we observe that
    \begin{equation}
        \text{if } pm\ppath{k} p'm' \text{ then } p(m+1)\ppath{k}p'(m'+1)\label{posprop}.
    \end{equation}
    After at most $K$ steps, our minimal path will repeat some control-state $p_j=p_l$ at positions
    $j<l<K$. By minimality and point \eqref{posprop} we can assume that $m_j<m_l$. Therefore,
    after at most $K$ such repetitions the counter will reach a value $\ge K$,
    and thus, by the first case above, 
    the remaining path must be of length $\le K$. This allows us to bound the length of the minimal
    path from $pm$ to control-state $q$ by $K^2+K$.

    Case 2: No minimal path witnessing $pm\reach qn$ is positive.
    Consider a minimal path $(p_0m_0), (p_1m_1),\dots,(p_km_k)$ from $pm=p_0m_0$ to $qn=p_km_k$
    and let $i_0,i_1,\dots,i_l$ be exactly those indices with $m_{i_j}=0$.
    We split the path into phases
    $p_0m_0\reach p_{i_0}m_{i_0}$, $p_{i_j}m_{i_j+1}\reach p_{i_{(j+1)}}m_{i_{(j+1})}$ for $0\le j<l$
    and $p_{i_{(l+1)}}m_{i_{(l+1)}}\reach p_km_k$ and consider the first, the
    last and the intermediate phases
    separately.

    {\it First phase.}
    The path $p_0m_0\reach p_{i_0}m_{i_0}$ can be split into parts 
    $q_j(m_0-j) \reach q_{j+1}(m_0-(j+1))$ for $0 \le j < m_0$ and 
    $q_0 = p_0$ and $q_{m_0} = p_{i_0}$, by considering the first occasions
    where the counter value reaches $m_0-j$. 
    In particular, inside the path $q_j(m_0-j) \reach q_{j+1}(m_0-(j+1))$
    the counter value does not drop below $m_0-j$ before the last step.
    By Lemma~\ref{1-0-reach} and point \eqref{posprop}
    there exists a path $q_j(m_0-j) \reach
    q_{j+1}(m_0-(j+1))$ of length $\le K^3$.
    Thus the path $p_0m_0\reach p_{i_0}m_{i_0}$ can be bounded by length $mK^3$.
    
    {\it Intermediate phases.}
    These are paths from some configuration $p_{i_j}0$ to $p_{i_{(j+1)}}0$. 
    Such a path is either of length $1$, or the first step increases the
    counter, i.e., $p_{i_j}0 \rightarrow q1$ for some control-state $q$.
    In the latter case, by minimality and Lemma \ref{1-0-reach} we can bound
    the path from $q1$ to $p_{i_{(j+1)}}0$ by $K^3$.
    Thus we can bound the path $p_{i_j}0$ to $p_{i_{(j+1)}}0$ by $K^3+1$.
    Note that there can only be at most $K$ such intermediate phases, because
    the path would otherwise repeat a configuration which would contradict its minimality.

    {\it Last phase.}
    The last phase is a positive path. Like in Case 1) we can bound its length by $K^2+K$.

    To conclude, the length of the shortest witness for $pm\reach qn$
    is bounded by $mK^3 +K(K^3+1) + K^2 +K  \le \max\{m,1\}5K^4$.
\end{proof}

\begin{theorem}
    Checking strong trace inclusion $T(pm)\subseteq T(q)$ 
    or weak trace inclusion $T(pm)\subseteqq T(q)$ for a \OCA\ process $pm$ and
    a finite process $q$ is PSPACE complete.
\end{theorem}
\begin{proof}
  A PSPACE lower bound holds already for strong trace inclusion of
  finite-state systems \cite{MS:IEEE1972}.
  The weak trace inclusion problem $T(pm)\subseteqq T(q)$ can trivially be
  reduced to the strong one by taking the transitive closure of the finite
  system w.r.t. invisible transitions.
  It remains to show a PSPACE upper bound for the problem $T(pm)\subseteq T(q)$.
  Let $pm$ be a configuration of the \OCA\ $A=(Q,Act,\delta,\delta_0)$ and $q$ a state of
  the NFA $B=(S,Act,\delta)$ and let $\bar{B}$ denote the powerset
  automaton of $B$.

  To check if $T(pm)\not\subseteq T(q)$ holds we can equivalently
  test $T(pm)\cap T(q)^c \neq \emptyset$. That is,
  if in the product automaton $A\x \bar{B}$ some control-state $(p',\emptyset)$ is reachable from initial
  configuration $(p,\{q\})m$.
  This can be checked by nondeterministically guessing a path stepwise.
  The finite control of the automaton $A\x \bar{B}$ is bounded by
  $K:=|Q|*2^{|S|}$.
  By Lemma~\ref{lem:bound} we know that the shortest path that witnesses such
  a control-state reachability is bounded by $B:=\max\{m,1\}5K^4$.
  This bounds the number of steps we need to consider until we can safely
  terminate and conclude that in fact trace inclusion holds.
  $B$ is polynomial in $m$ and $|Q|$ and exponential in $|S|$.
  However, we need only polynomial space to store a configuration
  of $A\x\bar{B}$ (with control-state numbers and counter values encoded in binary) 
  and the binary coded values of the search-depth and its
  bound $B$. Thus we can check the condition in PSPACE.
\end{proof}

\section{Summary and Conclusion}\label{sec:conclusion}

We summarize known results about the complexity of checking the following semantic
preorders/equivalences:
strong bisimulation $\sim$, weak bisimulation $\approx$,
strong simulation $\msim$, weak simulation $\mwsim$,
strong trace inclusion $\subseteq$ and weak trace inclusion $\subseteqq$.
In Table~\ref{tab:overview} we consider problems where systems of the same
type are compared, 
while in Table~\ref{tab:overview:fs} we consider the problems of checking
preorders/equivalences between infinite-state systems and finite-state
systems.

\begin{table}
  \caption{Decidability of preorders and equivalences on finite-state systems, \OCN\ and \OCA,
 resp. New results in boldface.}
  \label{tab:overview}
  \begin{tabular}{|l|c|c|c|}
    \hline
                      & FS                & \OCN                            & \OCA \\
    \hline
    $\sim$      & P-complete \cite{JS:SOFSEM2001}  & \multicolumn{2}{|c|} { PSPACE-complete \cite{Srb2009,BGJ2010}}\\
    \hline
    $\approx$   & P-complete \cite{JS:SOFSEM2001}  & \multicolumn{2}{|c|}{undecidable \cite{May2003}} \\
    \hline
    $\msim$     & P-complete \cite{JS:SOFSEM2001}  & decidable \cite{AC1998,JKM2000}  &    undecidable \cite{JMS1999}\\
                      &                   & PSPACE-hard \cite{Srb2009}     &   \\
    \hline
    $\mwsim$    & P-complete \cite{JS:SOFSEM2001}  & \bf{decidable}                 &    undecidable \cite{JMS1999}\\
    \hline
    $\subseteq/\subseteqq$ & PSPACE-compl. \cite{MS:IEEE1972}  &  \bf{undecidable}    &     undecidable \cite{Valiant1973}           \\
    \hline
  \end{tabular}
\end{table}

\begin{table}
  \caption{Known results on checking simulation, weak simulation and trace inclusion between one-counter
      and finite systems.}
  \label{tab:overview:fs}
  \begin{tabular}{|l|c|c|}
    \hline
                                & \OCN                                    &    \OCA                                    \\
    \hline
    $\sim {\it FS}$                      & P-complete \cite{Kuc2000}               &  P-complete \cite{Kuc2000} \\
    \hline
    $\approx {\it FS}$                   & $P^{\it NP}$-complete \cite{GMT2009}     & $P^{\it NP}$-complete \cite{GMT2009}  \\
    \hline
    $\msim {\it FS}$ (and ${\it FS}\msim$)  & P-complete \cite{Kuc2000a}              & PSPACE-complete\cite{Ser2006,Srb2009}  \\
\hline
    $\mwsim {\it FS}$ (and ${\it FS}\mwsim$)& \bf{P-complete}                        & PSPACE-complete\cite{Ser2006,Srb2009}  \\
\hline
$\subseteq/ \subseteqq {\it FS}$ & \bf{PSPACE-complete}                   & \bf{PSPACE-complete}\\
\hline
    ${\it FS}\subseteq/ \subseteqq $ & decidable \cite{JEM:JCSS1999}    & undecidable \cite{Valiant1973}   \\
\hline
  \end{tabular}
\end{table}

The construction used to show PSPACE hardness of strong bisimulation in \cite{Srb2009}
uses \OCN\ only, and moreover it can be modified to prove a PSPACE lower bound for checking
strong simulation between \OCA\ and finite systems (and vice-versa)
and strong simulation for \OCN; see Remark 3.8 in \cite{Srb2009}.

The proof of the undecidability of weak bisimulation between \OCN\ \cite{May2003}
can be modified to work even for the subclass of normed nets with unary alphabets.

A PSPACE upper bound for strong/weak simulation between \OCA\ and FS (and vice-versa)
can be obtained by reduction to $\mu$-calculus model checking
for \OCA, which is in PSPACE \cite{Ser2006}.

\newpage

\appendix
\section{Proof of Theorem \ref{thm:reduction}}\label{app:reduction}
{\bf\noindent Theorem~\ref{thm:reduction}.}
        For two one-counter nets $M$ and $N$ with states
        $Q_M$ and $Q_N$ resp., one
        can effectively construct a \OCN\ $M'$ with states $Q_{M'}\supseteq Q_M$
        and a $\omega$-net $N'$ with states $Q_{N'}\supseteq Q_N$ such that for
        each pair $p,q\in Q_M\x Q_N$ of original control states and
                any ordinal $\alpha$ the following hold.
        \begin{enumerate}
            \item $pm\mwsim qn$ w.r.t. $\!M,N$ iff $pm\msim\!qn$ w.r.t. $M',N'$.
            \item If $pm\WSIM{}{\alpha}\!qn$ w.r.t. $\!M,N$
                then $pm\SIM{}{\alpha}\!qn$ w.r.t. $M',N'$.
        \end{enumerate}

The reduction will be done in two steps.
First (Lemma \ref{lem:app_reduction:GON}) we reduce weak simulation for one-counter nets to strong simulation
beteween a one-counter net and yet another auxiliary model called \emph{guarded $\omega$-nets}.
These differ from $\omega$-nets in that each transition may change the counter by more than one and is
guarded by an integer, i.e. can only be applied if the current counter value exceeds the \emph{guard} attached to it.
In the second step (Lemma \ref{lem:app_reduction:normalize}) we normalize the effects of all transitions
to $\{-1,0,1,\omega\}$ and eliminate all integer guards and thereby construct an ordinary $\omega$-net
for \V.

Before we start observe that without loss of generality we can assume that every state $p$
allows a silent loop $p \step {\varepsilon, 0} p$.

\begin{definition}
    A \emph{path} in a one-counter net $N=(Q,\Act, \delta)$ is a sequence
    $\pi=(s_0,a_0,d_0,t_0)$ $(s_1,a_1,d_1,t_1)$ $\dots(s_k,a_k,d_k,t_k)\in\delta^*$
    of transitions where $s_{i+1}=t_i$ for all $i<k$.
    We call $\pi$ \emph{cyclic} if
    $s_i=t_j$ for some $0\le i<j\le k$ and write ${}^i\pi$ for its prefix of length $i$.
    A cyclic path is a \emph{loop} if $p_i\neq p_j$ for all $0\le i<j<k$.
    Define the \emph{effect} $\Delta(\pi)$ and \emph{guard} $\Gamma(\pi)$ of a path $\pi$ by
    \begin{align*}
        \Delta(\pi) = \sum_{i=0}^k d_i\text{\quad and }
        &&\Gamma(\pi) = - \min\{\Delta({}^i\pi)|i\le k\}
    \end{align*}
    where $n<\omega$ and $n+\omega=\omega+n=\omega$ for every $n\in\N$.
    The guard $\Gamma(\pi)$ denotes the minimal counter value that is needed to traverse the path $\pi$ while
    maintaining a non-negative counter value along all intermediate configurations.
    Lastly, fix a homomorphism $obs:\delta^*\to (Act\setminus\{\tau\})^*$, that maps paths to their \emph{observable
    action sequences}: $obs((s,\tau,d,t))=\eps$ and $obs((s,a,d,t))=a$ for $a\neq\tau$.
\end{definition}

\begin{definition}[Guarded $\omega$-Nets]
    A \emph{guarded $\omega$-net} $N=(Q,Act,\delta)$ is given by finite sets $Q,Act$
    of states and actions and a transition relation $\delta\subseteq Q\x Act\x\N\x\Z\cup\{\omega\}\x Q$. 
    It defines a transition system over the stateset $Q\x\N$ where $pm\step{a}qn$ iff
    there is a transition $(p,a,g,d,q)\in \delta$ with
    \begin{enumerate}
        \item $m\ge g$ and
        \item $n=m+d\in\N$ or $d=\omega$ and $n>m$.
    \end{enumerate}
\end{definition}

Specifically, $N$ is a \emph{$\omega$-net} if for all transitions $g=0$ and $d\in\{-1,0,1,\omega\}$.
The next construction establishes the connection between weak similarity of one-counter nets
and strong similarity between OCN and guarded $\omega$-net processes. 

\begin{lemma}
    \label{L1}
    For a one-counter net $N=(Q,Act,\delta)$ we can effectively construct a guarded $\omega$-net 
    $G=(Q,Act,\delta')$ such that for all $a\in Act$,
    \begin{enumerate}
        \item whenever $pm\Wstep{a}{}{N}qn$, there is a $n'\ge n$ such that $pm\Step{a}{}{G}qn'$
        \item whenever $pm\Step{a}{}{G}qn$, there is a $n'\ge n$ such that $pm\Wstep{a}{}{N}qn'.$
    \end{enumerate}
\end{lemma}

\begin{proof}
    The idea of the proof is to introduce direct transitions
    from one state to another for any path between them
    that reads at most one visible action and does not contain silent cycles.

    For two states $s,t$ of $N$, let $D(s,t)$ be the set of \emph{direct} paths 
    from $s$ to $t$:
    \begin{align*}
      D(s,t) = \{&(p_i,a_i,d_i,p_{i+1})_{i<k} : p_0=s, p_k=t,\\
                 &\forall_{0\le i<j\le k} p_i=p_j\implies (i=0\land j=k)\}.
    \end{align*}
    Define the subset of \emph{silent direct} paths by $SD(s,t) = \{\pi\in D(s,t) | obs(\pi)=\varepsilon\}$.
    Every path in $D(s,t)$ has acyclic prefixes only and is therefore bounded in length by $|Q|$.
    Hence $D(s,t)$ and $SD(s,t)$ are finite and effectively computable for all pairs $s,t$.
    
    Using this notation, we define the transitions in $G$ as follows.
    Let $\delta'$ contain a transition $(p,a,\Gamma(\pi),\Delta(\pi),q)$ for each path $\pi=\pi_1(s,a,d,s')\pi_2\in\delta^+$
    where $\pi_1\in SD(p,s)$ and $\pi_2\in SD(s',q)$.  This carries over all transitions of $N$ because the empty path
    is in $SD(s,s)$ for all states $s$. Moreover, introduce $\omega$-transitions in case $N$ allows paths $\pi_1,\pi_2$
    as above to contain direct cycles with positive effect on the counter: If there is a path
    $\pi=\pi_1'\pi_1''\pi_1'''(s,a,d,s')\pi_2$ with
    \begin{enumerate}
        \item $\pi_1'\in SD(p,t)$, $\pi_1''\in SD(t,t)$ and $\pi_1'''\in SD(t,s)$
        \item $\Delta(\pi_1'')>0$
    \end{enumerate}
    for some $t\in Q$, then $\delta'$ contains a transition $(p,a,\Gamma(\pi_1'\pi_1''),\omega,q)$.
    Similarly, if for some $t\in Q$, there is a path $\pi=\pi_1(s,a,d,s')\pi_2'\pi_2''\pi_2'''$ that satisfies
    \begin{enumerate}
        \item $\pi_1\in SD(p,s)$,$\pi_2'\in SD(s',t)$, $\pi_2''\in SD(t,t)$ and $\pi_2'''\in SD(t,q)$
        \item $\Delta(\pi_2'')>0$
    \end{enumerate}
    add a transition $(p,a,g,\omega,q)$ with guard $g=\Gamma(\pi_1(s,a,d,s')\pi_2'\pi_2'')$.
    If there is an $a$-labelled path from $p$ to $q$ that contains a silent and direct cycle with positive effect,
    $G$ has an a-labelled $\omega$-transition from $p$ to $q$ with the guard derived from that path.

    To prove the first part of the claim, assume $pm\Wstep{a}{}{N}qn$. By definition of weak steps, there must be a path
    $\pi=\pi_1(s,a,d,s')\pi_2$ with $obs(\pi_1)=obs(\pi_2)=\varepsilon$. Suppose both $\pi_1$ and
    $\pi_2$ do not contain loops with positive effect. Then there must be paths $\pi_1'\in SD(p,s), \pi_2'\in SD(s',q)$ with
    $\Gamma(\pi_i')\le\Gamma(\pi_i)$ and $\Delta(\pi_i')\ge\Delta(\pi_i)$ for $i\in\{1,2\}$ that can be obtained from
    $\pi_1$ and $\pi_2$ by removing all loops with effects less or equal $0$. So $G$ contains a transition
    $(p,a,g',d',q)$ for some $g'\le m$ and $d'\ge n-m$ and hence
    $pm\Step{a}{}{G}qn'$ for $n'=m+d'\ge n$. Alternatively, either $\pi_1$ or $\pi_2$ contains a
    loop with positive effect. Note that for any such path, another path with lower or equal guard exists that connects the
    same states and contains only one counter-increasing loop:
    If $\pi_1$ contains a loop with positive effect, there is a path $\bar{\pi_1}=\pi_1'\pi_1''\pi_1'''$ from
    $p$ to $s$, where $\pi_1',\pi''$ and $\pi_1'''$ are direct and $\Delta(\pi_1'')>0$ for the loop $\pi_1''\in SD(t,t)$
    for some state $t$. In this case, $G$ contains a $\omega$-transition $(p,a,g,\omega,q)$ with
    $g=\Gamma(\pi_1'\pi_1'')$. Similarly, if $\pi_2$ contains the counter-increasing loop, there is a
    $\bar{\pi_2}=\pi_2'\pi_2''\pi_2'''$, with $\pi_2'\in SD(s',t), \pi_2''\in SD(t,t), \pi_2'''\in SD(t,q)$ and
    $\Delta(\pi_2'')>0$. This means there is a transition $(p,a,g,\omega,q)$ in $G$ with
    $g=\Gamma(\pi_1(s,a,d,s')\pi_2'\pi_2'')$. In both cases, $g\le\Gamma(\pi)\le m$ and therefore
    $pm\Step{a}{}{G}qi$
    for all $i\ge m$.

    For the second part of the claim, assume $pm\Step{a}{}{G}qn$. This
    must be the result of a transition $(p,a,g,d,q)\in\delta'$ for some $g\le m$. 
    In case $d\neq\omega$, there is a path $\pi\in\delta^*$ from $p$ to $q$ with
    $\Delta(\pi)=n-m$, $obs(\pi)=a$ and $\Gamma(\pi)=g$ that witnesses
    the weak step $pm\Wstep{a}{}{N}qn$ in $N$. Otherwise if $d=\omega$, there must
    be a path $\pi=\pi_{11}\pi_{12}\pi_{13}(s,a,d,s')\pi_{21}\pi_{22}\pi_{23}$ from $p$ to $q$ in $N$
    where $\Gamma(\pi)\le m$, all $\pi_{ij}$ are silent and direct and one of $\pi_{12}$ and $\pi_{22}$ 
    is a cycle with
    strictly positive effect.
    This implies that one can ``pump'' the value of the counter higher than any given value.
    Specifically, there are naturals $k$ and $j$ such that the path $\pi'=\pi_{11}\pi_{12}^k\pi_{13}(s,a,d,s')
    \pi_{21}\pi_{22}^j\pi_{23}$ from $p$ to $q$
    satisfies $\Gamma(\pi')\le\Gamma(\pi)\le m$ and $\Delta(\pi')\ge m-n$. Now $\pi'$
    witnesses the weak step $pm\Wstep{a}{}{N}qn'$ in $N$ for an $n'\ge n$.
\end{proof}

\begin{remark}
     Observe that no transition of the net $G$ as constructed above has a guard larger than $|Q|*3+1$
     and finite effect $> 2|Q|+1$.
\end{remark}
\begin{lemma}
    \label{lem:app_reduction:GON}
    For a one-counter net $N=(Q,Act,\delta)$ one can effectively construct a
    guarded $\omega$-net $G=(Q,Act,\delta')$ s.t. for any \OCN\ $M$ and
    any two configurations $pm,qn$ of $M$ and $N$ resp.,
    \begin{equation}
    pm\mwsim qn\text{ w.r.t. }M,N \iff pm\msim qn\text{ w.r.t. }M,G.\label{lem:app_reduction:GON:claim}
    \end{equation}
\end{lemma}
\begin{proof}
    Consider the construction from the proof of Lemma \ref{L1}.
    Let $\WSIM{}{M,N}$ be the largest weak simulation w.r.t.\ $M,N$ and
    $\SIM{}{M,G}$ be the largest strong simulation w.r.t.\ $M,G$.
    
    For the ``if'' direction we show that $\SIM{}{M,G}$ is a weak simulation w.r.t.\ $M,N$. Assume
    $pm\SIM{}{M,G}qn$ and $pm\Step{a}{}{M} p'm'$. That means there is a step $qn\Step{a}{}{G} q'n'$ for some
    $n'\in N$ so that $p'm'\SIM{}{G} q'n'$.
    By Lemma~\ref{L1} part 2, $qn\Wstep{a}{}{N} q'n''$ for a $n''\ge n'$. Because simulation is monotonic
    we know that also $p'm'\SIM{}{M,G} q'n''$. Similarly, for the ``only if''
    direction, one can use the first claim of Lemma \ref{L1} to check that
    $\WSIM{}{M,N}$ is a strong simulation w.r.t.\ $M,G$.\qed
\end{proof}

\begin{lemma}
    \label{lem:app_reduction:normalize}
    For a one-counter net $M$
    and a guarded $\omega$-net $G$
    one can effectively construct one-counter nets $M',G'$
    such that
    for any two configurations $pm,qn$ of $M$ and $G$ resp.,
    \begin{equation}
    pm\msim qn\text{ w.r.t. }M,G \iff pm\msim qn\text{ w.r.t. }M',G'.\label{lem:app_reduction:normalize:claim}
    \end{equation}
\end{lemma}
\begin{proof}
     We first observe that for any transition of the guarded $\omega$-net $G$,
     the values of its guard is bounded by some constant. The same holds for
     all finite effects. Let $\Gamma(G)$ be the maximal guard and $\Delta(G)$
     be the maximal absolute finite effect of any transition of $G$.

     The idea of this construction is to simulate one round of the game $M$ vs. $G$
     in $k=2\Gamma(G)+\Delta(G)+1$ rounds of a simulation game $M'$ vs.\ $G'$.
     We will replace original steps of both players by sequences of $k$ steps in the new game,
     which is long enough to verify if the guard of \V's move is satisfied and
     adjust the counter using transitions with effects in $\{-1,0,+1,\omega\}$ only.

     We transform the net $M=(Q_M,\Act,\delta_M)$ to the net $M'=(Q_{M'},\Act',\delta_{M'})$ as
     follows:
     \begin{align}
         \Act' =\ &\Act \cup \{b\}\\
          Q_{M'} =\ &Q_M\cup \{p_i|1\le i<k, p\in Q_M\}\\
     \delta_{M'} =\ &\{p\step{a,d}q_k| p\step{a,d}q\in\delta_M\}\\
                   &\cup \{p_i\step{b,0}p_{i-1}|1<i<k\}\\
                   &\cup \{p_1\step{b,0}q\}.
     \end{align}
     We see that
     \begin{equation}
         pm\Step{a}{}{M}qn' \iff pm\Step{a}{}{M'}q_{k-1}n'\Step{b^{k-2}}{}{M'}q_1n'\Step{b}{}{M'}qn'.\label{spoilernet}
     \end{equation}
     
     Now we transform the guarded $\omega$-net $G=(Q_G,\Act,\delta_G)$ to the $\omega$-net
     $G'=(Q_{G'},\Act',\delta_{G'})$.
     Every original transition will be replaced by a sequence of $k$ steps
     that test if the current counter value exceeds the guard $g$
     and adjust the counter accordingly.
     The new net $G'$ has states
     \begin{equation}
          Q_{G'} = Q_G \cup \{t_i|0\le i<k, t\in \delta_G\}.
     \end{equation}
     For each original transition $t=(p,a,g,d,q)\in\delta_G$, we add the following transitions
     to $\delta_{G'}$. First, to test the guard:
     \begin{align}
          p\step{a,0}t_{k-1},\\
          t_i\step{b,-1}t_{i-1}, &\text{ for } k-g<i<k\\
          t_i\step{b,+1}t_{i-1}, &\text{ for } k-2g<i<k-g.
     \end{align}
     Now we add transitions to adjust the counter according to $d\in\N\cup\{\omega\}$.
     In case $0\le d<\omega$ we add
     \begin{align}
          t_i\step{b,+1}t_{i-1}, &\text{ for } k-2g-|d|<i<k-2g\\
          t_i\step{b,0}t_{i-1}, &\text{ for } 0\le i<k-2g-d.
     \end{align}
     In case $d<0$ we add
     \begin{align}
          t_i\step{b,-1}t_{i-1}, &\text{ for } k-2g-|d|<i<k-2g\\
          t_i\step{b,0}t_{i-1}, &\text{ for } 0\le i<k-2g+d.
     \end{align}
     In case $d=\omega$ we add
     \begin{align}
          t_i\step{b,\omega}t_{i-1}, &\text{ for } i=k-2g\\
          t_i\step{b,0}t_{i-1}, &\text{ for } 0\le i<k-2g.
     \end{align}
     Finally, we allow a move to the new state:
     \begin{equation}
          t_0\step{b,0}q.
     \end{equation}
     Observe that every transition in the constructed net $G'$
     has effect in $\{-1,0,+1,\omega\}$. $G'$ is therefore an ordinary $\omega$-net.
     It is streightforward to see that
     \begin{equation}
         pm\Step{a}{}{G}qn' \iff pm\Step{ab^{k-1}}{}{G'}qn'.\label{dupnet}
     \end{equation}
     The claim \eqref{lem:app_reduction:normalize:claim} now follows from Equations \eqref{spoilernet}
     and \eqref{dupnet}.
     This conludes the proof of Lemma \ref{lem:app_reduction:normalize} and
     point \ref{thm:reduction:main} of Theorem~\ref{thm:reduction}.

     For point \ref{thm:reduction:ordinals} of the claim observe that by construction of $M'$ and
     $N'$, one round of a weak simulation game w.r.t. $M,N$ is simulated by $k$ rounds of a
     simulation game w.r.t. $M',N'$.
     Therefore, if if \R\ has a strategy to win the simulation game w.r.t. $M', N'$ in $\alpha$
     rounds then he can derive a strategy to win the game w.r.t. $M,N$ in not more than $\alpha$
     rounds. So if $pm\notSIM{}{\alpha}qn$ w.r.t. $M',N'$ then $pm\notWSIM{}{\alpha}\!qn$ w.r.t.
     $M,N$.
\end{proof}

\end{document}